\documentclass[11pt, reqno]{amsart}

\usepackage{amssymb}
\usepackage{amsmath}
\usepackage{graphicx}
\usepackage{enumerate}
\usepackage{fullpage}
\usepackage{pst-all}

\def\stackunder#1#2{\mathrel{\mathop{#2}\limits_{#1}}}

\makeatletter\@addtoreset{equation}{section}\makeatother

\allowdisplaybreaks

\newcommand{\nabgam}{\nabla^{\Gamma}_{\gamma}}

\begin{document}

\theoremstyle{plain}
\newtheorem{theorem}{Theorem}[section]
\newtheorem{proposition}[theorem]{Proposition}
\newtheorem{corollary}[theorem]{Corollary}
\newtheorem{condition}[theorem]{Condition}
\newtheorem{example}[theorem]{Example}
\theoremstyle{definition}
\newtheorem{definition}[theorem]{Definition}
\theoremstyle{remark}
\newtheorem{lemma}[theorem]{Lemma}
\newtheorem{remark}[theorem]{Remark}
\newtheorem{notation}[theorem]{Notation}

\newcommand{\grad}{\nabla}
\newcommand{\D}{\partial}
\newcommand{\E}{\mathcal{E}}
\newcommand{\N}{\mathbb{N}}
\newcommand{\dom}{\mathcal{D}}
\newcommand{\ess}{\operatorname{ess~inf}}

\title[Tagged Particle Process]{Tagged Particle Process in Continuum with Singular Interactions}

\author{Torben Fattler, Martin Grothaus}

\address{
Torben Fattler, Mathematics Department, University of
Kaiserslautern, \newline P.O.Box 3049, 67653 Kaiserslautern,
Germany. {\rm \texttt{Email:~fattler@mathematik.uni-kl.de}, \newline
\texttt{URL:~http://www.mathematik.uni-kl.de/~wwwfktn/homepage/fattler.html}}
\newline
Martin Grothaus, Mathematics Department, University of
Kaiserslautern, \newline P.O.Box 3049, 67653 Kaiserslautern,
Germany. 
\newline {\rm \texttt{Email:~grothaus@mathematik.uni-kl.de},
\texttt{URL:~http://www.mathematik.uni-kl.de/$\sim$grothaus/ }} }

\date{\today}

\thanks{
\\
We thank Yuri Kondratiev, Michael R\"ockner, Sven Struckmeier and Heinrich v.~Weizs\"acker for discussions and helpful comments. Special thanks go to Florian Conrad who proposed a proof for conservativity of the coupled process. Financial support by the DFG through the project GR 1809/4-2 is gratefully
acknowledged.}


\begin{abstract}
We study the dynamics of a tagged particle in an infinite particle environment. Such processes have been studied in e.g.~\cite{GP85}, \cite{DeM89} and \cite{Os98}. I.e., we consider the heuristic system of stochastic differential equations
\begin{align}
&d\xi(t)=\sum_{i=1}^\infty\nabla\phi(y_i(t))\,dt+\sqrt{2}\,dB_1(t),\quad t\ge 0,\tag{TP}\\
&\left.
\begin{matrix}
dy_i(t)=\big(-\sum_{\stackunder{j\not=i}{j=1}}^\infty\nabla\phi(y_i(t)-y_j(t))-\nabla\phi(y_i(t))-\sum_{j=1}^\infty\nabla\phi(y_j(t))\big)\,dt\\
+\sqrt{2}\,d(B_{i+1}(t)-B_1(t)),\quad t\ge 0,\quad i\in\mathbb{N}   
\end{matrix}\right\}.\tag{ENV}
\end{align}
This system realizes the coupling of the motion of the tagged particle, described by (TP), and the motion of the environment seen from the tagged particle, described by (ENV).
As we can observe in (TP) the solution to (ENV), the so-called environment process, is driving the tagged particle. Thus our strategy is to study (ENV) at first and afterwards the coupled process, i.e., (TP) and (ENV) simultaneously.
Here the analysis and geometry on configuration spaces developed in \cite{AKR98a} and \cite{AKR98b} plays an important role. Furthermore, the harmonic analysis on configuration spaces derived in \cite{KK99a} is very useful for our
considerations. First we derive an integration by parts formula with respect to the standard gradient $\nabla^\Gamma$ on configuration spaces $\Gamma$ for a general class of grand canonical Gibbs measures $\mu$, corresponding to pair potentials $\phi$ and intensity measures $\sigma=z\,\exp(-\phi)\,dx,~0<z<\infty$, having correlation functions fulfilling a Ruelle bound. Furthermore, we use a second integration by parts formula with respect to the gradient $\nabgam$, generating the uniform translations on $\Gamma$, for a (non-empty) subclass of the Gibbs measures $\mu$ as above which is provided in \cite{CoKu09}. Combining these two gradients by Dirichlet form techniques we can construct the environment process and the coupled process, respectively. Scaling limits of such dynamics have been studied e.g.~in \cite{DeM89}, \cite{GP85} and \cite{Os98}. Our results give the first mathematically rigorous and complete construction of the tagged particle process in continuum with interaction potential. In particular, we can treat interaction potentials which might have a singularity at the origin, non-trivial negative part and infinite range as e.g.~the Lennard--Jones potential. 
  
\end{abstract}

\maketitle

\section{Introduction}
   
We consider a system of infinitely many Brownian particles in $\mathbb{R}^d,~d\in\mathbb{N}$, interacting via the gradient of a symmetric pair potential $\phi$. Since each particle can move through each position in space, the system is called continuous and is used for modeling suspensions, gases or fluids. 
The infinite volume, infinite particle, stochastic dynamics $(x(t))_{t\ge 0}$ heuristically solves the following infinite system of stochastic differential equations:
\begin{align}\label{sde}
dx_i(t)&=-\sum_{\stackunder{j\not=i}{j=1}}^\infty\nabla\phi(x_i(t)-x_j(t))dt+\sqrt{2}dB_i(t),\quad t\ge 0,\quad i\in\mathbb{N},
\end{align}
where $x(t)=\{x_1(t),x_2(t),\ldots\}$,~$t\ge 0$, and $(B_i)_{i\in\mathbb{N}}$ is a \emph{sequence of independent Brownian motions}. Its informal generator is given by
\begin{align}\label{genivipdif}
L_{\scriptscriptstyle{gsd}}=\sum_{i=1}^\infty\partial_{x_i}^2-\sum_{i=1}^\infty\Big(\sum_{\stackunder{j\not=i}{j=1}}^\infty\nabla\phi(x_i-x_j)\Big)\partial_{x_i}.
\end{align}
Using
\begin{align}\label{densityivipdif}
\varrho_{\scriptscriptstyle{\infty}}(x_1,x_2,\ldots)=\exp\Big(-\frac{1}{2}\sum_{i\not=j}\phi(x_i-x_j)\Big)
\end{align}
we have
\begin{align*}
L_{\scriptscriptstyle{gsd}}=\sum_{i=1}^\infty\partial_{x_i}\ln(\varrho_{\scriptscriptstyle{\infty}})\partial_{x_i}+\sum_{i=1}^\infty\partial^2_{x_i}.
\end{align*}
Note that $L_{\scriptscriptstyle{gsd}}$ in this form is not well-defined.
The construction of such diffusions has been initiated by R.~Lang \cite{La77}, who considered the case $\phi\in C^3_0(\mathbb{R}^d)$ using finite dimensional approximations and stochastic differential equations. More singular $\phi$, which are of particular interest in physics, as e.g.~the Lennard--Jones potential, have been treated by H.~Osada, \cite{Os96}, and M.~Yoshida, \cite{Y96}. Osada and Yoshida were the first to use Dirichlet forms for the construction of such processes. However, they could not write down the corresponding generators or martingale problems explicitly, hence could not prove that their processes actually solve (\ref{sde}) weakly. This, however, was proved in \cite{AKR98b} by showing an integration by parts formula for the respective grand canonical Gibbs measures. Another approach not using an integration by parts can be found in \cite{MaRo00}. In \cite{GKR04} the authors provide an $N/V$-limit for the infinite volume, infinite particle stochastic dynamics with singular interactions in continuous particle systems on $\mathbb{R}^d,~d\ge 1$.  Their construction is the first covering the case $d=1$ in the space of single configurations (only one particle at one site for all times $t\ge 0$).    

In this paper we study the tagged particle process in continuum with singular interactions. The underlying model can be described as follows. Consider the infinite system of Brownian particles described by (\ref{sde}). Coloring any one particle from the system blue and all the rest of the particles yellow, we investigate the motion of this \emph{tagged particle} in the random sea of all the yellow ones.
In \cite{GP85} this model and a scaling limit of it is studied for Brownian particles in $\mathbb{R}^d$ interacting via the gradient of a smooth, finite range, symmetric, positive pair potential. In \cite{Os98} the author considers the tagged particle process for more singular potentials, including the Lennard--Jones potential, using Dirichlet form techniques. However, there the author is also mainly interested in obtaining a scaling limit for the tagged particle process. Showing existence of the stochastic dynamics in the above cited articles has been left open. Osada gives reference to a forthcoming paper on his own, but as far as we know, it has never been published. Thus in our opinion there is a need to construct the tagged particle process with interaction potential rigorously. For other strategies to obtain the tagged particle process see e.g. \cite[Sect.~6]{DeM89} and the references therein. But note that these are not worked out in detail. We start with an heuristic approach just to clarify the way of posing the problem. After doing so the whole analysis will be done on a strictly rigorous level. 

Assume we are given a solution $x(t),~t\ge 0$, of (\ref{sde}). Using the coordinate transformation 
\begin{align}\label{ctrans}
&\xi(t):=x_1(t)\quad\mbox{and}\nonumber\\ 
&y_i(t):=x_{i+1}(t)-x_1(t),~i\in\mathbb{N},
\end{align}
we can rewrite (\ref{sde}) and obtain
\begin{align}
&d\xi(t)=\sum_{i=1}^\infty\nabla\phi(y_i(t))\,dt+\sqrt{2}\,dB_1(t)\label{tppro}\\
&\left.
\begin{matrix}
dy_i(t)=\big(-\sum_{\stackunder{j\not=i}{j=1}}^\infty\nabla\phi(y_i(t)-y_j(t))-\nabla\phi(y_i(t))-\sum_{j=1}^\infty\nabla\phi(y_j(t))\big)\,dt\\
+\sqrt{2}\,d(B_{i+1}(t)-B_1(t)),\quad t\ge 0,\quad i\in\mathbb{N}   
\end{matrix}\right\}.\label{envpro}
\end{align}
To derive the informal generator\index{generator!informal} of the process $\{\xi(t),y_1(t),y_2(t),\ldots\}$ corresponding to (\ref{tppro}) and (\ref{envpro}), we use again the coordinate transformation (\ref{ctrans}) and obtain
\begin{align*}
&\partial_{x_1}=\partial_{\xi}-\sum_{i=1}^\infty\partial_{y_i},\\
&\partial_{x_{i+1}}=\partial_{y_i},\quad i\ge 1.
\end{align*}
Plugging this into the representation of $L_{\scriptscriptstyle{gsd}}$ in (\ref{genivipdif}) yields
\begin{align*}
L_{\scriptscriptstyle{\text{coup}}}=\sum_{i=1}^\infty\partial_{y_i}^2+\left(\partial_\xi-\sum_{i=1}^\infty\partial_{y_i}\right)^2+\sum_{i=1}^\infty\nabla\phi(y_i)\left(\partial_\xi-\sum_{i=1}^\infty\partial_{y_i}\right)\nonumber\\
-\sum_{i=1}^\infty\Big(\nabla\phi(y_i)+\sum_{\stackunder{j\not=i}{j=1}}^\infty\nabla\phi(y_i-y_j)\Big)\partial_{y_i}.
\end{align*}
By setting
\begin{align}\label{denscoup}
\hat{\varrho}_{\scriptscriptstyle{\infty}}(y_1,y_2,\ldots):=\exp\Big(-\frac{1}{2}\sum_{i\not=j}\phi(y_i-y_j)\underbrace{-\sum_{i=1}^\infty\phi(y_i)}_{\text{\tiny additional term}}\Big)
\end{align}
we obtain
\begin{multline*}
L_{\scriptscriptstyle{\text{coup}}}=\sum_{i=1}^\infty\partial_{y_i}\ln(\hat{\varrho}_{\scriptscriptstyle{\infty}})\partial_{y_i}+\sum_{i=1}^\infty\partial^2_{y_i}
+\partial_\xi^2-\partial_\xi\sum_{i=1}^\infty\partial_{y_i}-\sum_{i=1}^\infty\partial_{y_i}\ln(\hat{\varrho}_{\scriptscriptstyle{\infty}})\partial_{\xi}-\sum_{i=1}^\infty\partial_{y_i}\partial_\xi\\
+\sum_{i=1}^\infty\ln(\hat{\varrho}_{\scriptscriptstyle{\infty}})\sum_{i=1}^\infty\partial_{y_i}+\left(\sum_{i=1}^\infty\partial_{y_i}\right)^2.
\end{multline*}
Hence $L_{\scriptscriptstyle{\text{coup}}}$ splits into
\begin{align}\label{igencoup}
L_{\scriptscriptstyle{\text{coup}}}=L_{\scriptscriptstyle{\text{env}}}+\partial_\xi^2-\partial_\xi\sum_{i=1}^\infty\partial_{y_i}-\sum_{i=1}^\infty\partial_{y_i}\ln(\hat{\varrho}_{\scriptscriptstyle{\infty}})\partial_{\xi}-\sum_{i=1}^\infty\partial_{y_i}\partial_\xi,
\end{align}
where
\begin{align*}
L_{\scriptscriptstyle{\text{env}}}=\sum_{i=1}^\infty\partial_{y_i}^2+\sum_{i=1}^\infty\partial_{y_i}\ln(\hat{\varrho}_{\scriptscriptstyle{\infty}})\partial_{y_i}+\left(\sum_{i=1}^\infty\partial_{y_i}\right)^2+\sum_{i=1}^\infty\partial_{y_i}\ln(\hat{\varrho}_{\scriptscriptstyle{\infty}})\sum_{i=1}^\infty\partial_{y_i}.
\end{align*}
For $y(t):=\{y_1(t),y_2(t),\ldots\}$, $(y(t))_{t\ge 0}$ is called \emph{environment process}\index{environment process}. It is the marginal of the $(\xi,y)$-process describing the environment seen from the tagged particle and having $L_{\scriptscriptstyle{\text{env}}}$ as informal generator. $L_{\scriptscriptstyle{\text{env}}}$ can be written as
\begin{align}\label{igenenv}
L_{\scriptscriptstyle{\text{env}}}=L_{\scriptscriptstyle{\text{gsdad}}}+\left(\sum_{i=1}^\infty\partial_{y_i}\right)^2+\sum_{i=1}^\infty\partial_{y_i}\ln(\hat{\varrho}_{\scriptscriptstyle{\infty}})\sum_{i=1}^\infty\partial_{y_i},
\end{align}
where
\begin{align}\label{igenmgsd}
L_{\scriptscriptstyle{\text{gsdad}}}:=\sum_{i=1}^\infty\partial_{y_i}^2+\sum_{i=1}^\infty\partial_{y_i}\ln(\hat{\varrho}_{\scriptscriptstyle{\infty}})\partial_{y_i}
\end{align}
is the informal generator of a gradient stochastic dynamics\index{gradient stochastic dynamics} with additional drift term (compare (\ref{densityivipdif}) and (\ref{denscoup})). In the sequel we call the dynamics corresponding to $L_{\scriptscriptstyle{\text{gsdad}}}$ a \emph{gradient stochastic dynamics with additional drift}.  

Now (\ref{tppro}) with $\xi(0)=0$ describes the motion $\xi(t)$ of the tagged particle which is determined by the environment process $(y(t))_{t\ge 0}$. 
Thus $L_{\scriptscriptstyle{\text{coup}}}$ is the informal generator of the diffusion process coupling\index{coupling} the motion of the tagged particle in $\mathbb{R}^d$ and the motion of the environment seen from this particle. The tagged particle process is then obtained by a projection of the coupled process generated informally by $L_{\scriptscriptstyle{\text{coup}}}$. 

On a rigorous level the infinite volume, infinite particle, stochastic dynamics in continuous particle systems can be realized as an infinite dimensional diffusion process taking values in the configuration space
\begin{eqnarray*}
\Gamma := \left\{ \gamma \subset {\mathbb R}^{d} \big| \, |\gamma \cap K| < \infty \,
\,\, \mbox{for any compact} \, K \subset {\mathbb R}^{d} \right\}
\end{eqnarray*}
and having a grand canonical Gibbs measure as an invariant measure. In \cite{AKR98b} the generator realizing $L_{\scriptscriptstyle{gsd}}$ (see (\ref{genivipdif})) is given by
\begin{multline*}
L^{\scriptscriptstyle{\Gamma,\mu_{\scriptscriptstyle{0}}}}_{\scriptscriptstyle{\text{gsd}}} F(\gamma)=\sum_{i,j=1}^N\partial_i\partial_jg_{\scriptscriptstyle{F}}(\langle f_1,\gamma\rangle,\ldots,\langle f_N,\gamma\rangle)\left\langle(\nabla f_i,\nabla f_j)_{\mathbb{R}^d},\gamma\right\rangle\\
+\sum_{j=1}^N\partial_j g_{\scriptscriptstyle{F}}(\langle f_1,\gamma\rangle,\ldots,\langle f_N,\gamma\rangle)\Big(\langle\Delta f_j,\gamma\rangle
-\sum_{\{x,y\}\subset\gamma}(\nabla\phi(x-y),\nabla f_j(x)-\nabla f_j(y))_{\mathbb{R}^d}\Big)\\
\mbox{for }\mu_{\scriptscriptstyle{0}}\mbox{-a.e.~}\gamma\in\Gamma\mbox{ and }F=g_{\scriptscriptstyle{F}}(\langle f_1,\cdot\rangle,\ldots,\langle f_N,\cdot\rangle)\in\mathcal{F}C_b^\infty(C^\infty_0(\mathbb{R}^d),\Gamma).
\end{multline*}
It is obtained by carrying out an integration by parts\index{integration by parts} of
\begin{align*}
\mathcal{E}^{\scriptscriptstyle{\Gamma,{\mu}}_{\scriptscriptstyle{0}}}_{\scriptscriptstyle{\text{gsd}}}(F,G)=\int_\Gamma\left(\nabla^\Gamma F(\gamma),\nabla^\Gamma G(\gamma)\right)_{T_\gamma\Gamma}\,d\mu_{\scriptscriptstyle{0}}(\gamma),\quad F,G\in\mathcal{F}C_b^\infty(C^\infty_0(\mathbb{R}^d),\Gamma),
\end{align*}
with respect to a grand canonical Gibbs measure $\mu_{\scriptscriptstyle{0}}$ corresponding to an intensity measure $\sigma=z\,dx,~0<z<\infty$. 

We start our analysis by considering the operator realizing $L_{\scriptscriptstyle{\text{gsdad}}}$ (see (\ref{igenmgsd})). It is given by
\begin{multline*}
L^{\scriptscriptstyle{\Gamma,\mu}}_{\scriptscriptstyle{\text{gsdad}}} F(\gamma)=\sum_{i,j=1}^N\partial_i\partial_jg_{\scriptscriptstyle{F}}(\langle f_1,\gamma\rangle,\ldots,\langle f_N,\gamma\rangle)\left\langle(\nabla f_i,\nabla f_j)_{\mathbb{R}^d},\gamma\right\rangle\\
+\sum_{j=1}^N\partial_j g_{\scriptscriptstyle{F}}(\langle f_1,\gamma\rangle,\ldots,\langle f_N,\gamma\rangle)\Big(\langle\Delta f_j,\gamma\rangle+\sum_{x\in\gamma}(\nabla\phi(x),\nabla f_j(x))_{\mathbb{R}^d}
\\-\sum_{\{x,y\}\subset\gamma}(\nabla\phi(x-y),\nabla f_j(x)-\nabla f_j(y))_{\mathbb{R}^d}\Big)\\
\mbox{for }\mu\mbox{-a.e.~}\gamma\in\Gamma\mbox{ and }F=g_{\scriptscriptstyle{F}}(\langle f_1,\cdot\rangle,\ldots,\langle f_N,\cdot\rangle)\in\mathcal{F}C_b^\infty(C^\infty_0(\mathbb{R}^d),\Gamma), 
\end{multline*}
where $\mu\in\mathcal{G}_{\scriptscriptstyle{Rb}}^{\scriptscriptstyle{gc}}(\Phi_{\scriptscriptstyle{\phi}},z\exp(-\phi))$, i.e.~a grand canonical Gibbs measure corresponding to a pair potential $\phi$ and an intensity measure $\sigma=z\,\exp(-\phi) dx,~0<z<\infty$, with corresponding correlation measures fulfilling a Ruelle bound. The associated symmetric bilinear form is given by
\begin{align*}
\mathcal{E}^{\scriptscriptstyle{\Gamma,\mu}}_{\scriptscriptstyle{\text{gsdad}}}(F,G)=\int_\Gamma\left(\nabla^\Gamma F(\gamma),\nabla^\Gamma G(\gamma)\right)_{T_\gamma\Gamma}\,d\mu(\gamma),\quad F,G\in\mathcal{F}C_b^\infty(C^\infty_0(\mathbb{R}^d),\Gamma).
\end{align*}
Having a Ruelle bound enables us to prove an integration by parts formula for cylinder functions on the configuration space $\Gamma$ with respect to the underlying grand canonical Gibbs measure $\mu$ for a general class of pair potentials $\phi$. This is done in Section \ref{secintbp1}, see Theorem \ref{thmintbyparts}. Using this result we can identify $L^{\scriptscriptstyle{\Gamma,\mu}}_{\scriptscriptstyle{\text{gsdad}}}$ as generator of $\mathcal{E}^{\scriptscriptstyle{\Gamma,\mu}}_{\scriptscriptstyle{\text{gsdad}}}$ on $\mathcal{F}C_b^\infty(C^\infty_0(\mathbb{R}^d),\Gamma)$. Moreover, showing that $(\mathcal{E}^{\scriptscriptstyle{\Gamma,\mu}}_{\scriptscriptstyle{\text{gsdad}}},D(\mathcal{E}^{\scriptscriptstyle{\Gamma,\mu}}_{\scriptscriptstyle{\text{gsdad}}}))$ is a conservative, local, quasi-regular, symmetric Dirichlet form we have the existence of a conservative diffusion process $\mathbf{M}_{\scriptscriptstyle{\text{gsdad}}}^{\scriptscriptstyle{\Gamma,\mu}}$ solving the associated martingale problem. To tackle $\mathcal{E}^{\scriptscriptstyle{\Gamma,\mu}}_{\scriptscriptstyle{\text{gsdad}}}$ we got many ideas from \cite{AKR98b}, but due to the more general intensity measure $\sigma$ according to $\mu$, we have to deal with additional technical problems.

Next step is to investigate the operator realizing $L_{\scriptscriptstyle{\text{env}}}$ (see (\ref{igenenv})). Hence we consider
\begin{multline*}
L^{\scriptscriptstyle{\Gamma,\mu}}_{\scriptscriptstyle{\text{env}}}F(\gamma)=L^{\scriptscriptstyle{\Gamma,\mu}}_{\scriptscriptstyle{\text{gsdad}}}F(\gamma)
+\sum_{i,j=1}^N\partial_i\partial_j g_{\scriptscriptstyle{F}}\left(\langle f_1,\gamma\rangle,\ldots,\langle f_N,\gamma\rangle\right)\big(\langle\nabla f_i,\gamma\rangle,\langle\nabla f_j,\gamma\rangle\big)_{\mathbb{R}^d}\\
+\sum_{j=1}^N \partial_j g_{\scriptscriptstyle{F}}\left(\langle f_1,\gamma\rangle,\ldots,\langle f_N,\gamma\rangle\right)\Big(\langle\Delta f_j,\gamma\rangle-\big(\langle\nabla\phi,\gamma\rangle,\langle\nabla f_j,\gamma\rangle\big)_{\mathbb{R}^d}\Big)\quad\mbox{ for }\mu\mbox{-a.e.~}\gamma\in\Gamma,
\end{multline*}
$F=g_{\scriptscriptstyle{F}}(\langle f_1,\cdot\rangle,\ldots,\langle f_N,\cdot\rangle)\in\mathcal{F}C_b^\infty(C^\infty_0(\mathbb{R}^d),\Gamma)$ and the associated symmetric bilinear form
\begin{multline*}
\mathcal{E}^{\scriptscriptstyle{\Gamma,{\mu}}}_{\scriptscriptstyle{\text{env}}}(F,G)=\int_\Gamma\Big(\nabla^\Gamma F(\gamma),\nabla^\Gamma G(\gamma)\Big)_{T_\gamma\Gamma}+\Big(\nabgam F(\gamma),\nabgam G(\gamma)\Big)_{\scriptscriptstyle{\mathbb{R}^d}}\,d{\mu}(\gamma),\\ F,G\in\mathcal{F}C_b^\infty(C^\infty_0(\mathbb{R}^d),\Gamma).
\end{multline*}
Here $\mu$ is again as above. For an activity $0<z<\infty$ and a general class of pair potentials $\phi$ in \cite{CoKu09} for a non-empty subset of $\mathcal{G}_{\scriptscriptstyle{Rb}}^{\scriptscriptstyle{gc}}(\Phi_{\scriptscriptstyle{\phi}},z\exp(-\phi))$ an integration by parts formula with respect to $\nabgam$ is shown. In the sequel we denote this subset of $\mathcal{G}_{\scriptscriptstyle{Rb}}^{\scriptscriptstyle{gc}}(\Phi_{\scriptscriptstyle{\phi}},z\exp(-\phi))$ by $\mathcal{G}_{\scriptscriptstyle{ibp}}^{\scriptscriptstyle{gc}}(\Phi_{\scriptscriptstyle{\phi}},z\exp(-\phi))$. Hence for $\mu\in \mathcal{G}_{\scriptscriptstyle{ibp}}^{\scriptscriptstyle{gc}}(\Phi_{\scriptscriptstyle{\phi}},z\exp(-\phi))$ the bilinear form $(\mathcal{E}^{\scriptscriptstyle{\Gamma,{\mu}}}_{\scriptscriptstyle{\text{env}}},D(\mathcal{E}^{\scriptscriptstyle{\Gamma,{\mu}}}_{\scriptscriptstyle{\text{env}}}))$ is closable. Furthermore, together with the results we obtained for $(\mathcal{E}^{\scriptscriptstyle{\Gamma,\mu}}_{\scriptscriptstyle{\text{gsdad}}},D(\mathcal{E}^{\scriptscriptstyle{\Gamma,\mu}}_{\scriptscriptstyle{\text{gsdad}}}))$ we prove that $(\mathcal{E}^{\scriptscriptstyle{\Gamma,\mu}}_{\scriptscriptstyle{\text{env}}},D(\mathcal{E}^{\scriptscriptstyle{\Gamma,\mu}}_{\scriptscriptstyle{\text{env}}}))$ is a conservative, local, quasi-regular, symmetric Dirichlet form. Thus we obtain a conservative diffusion process $\mathbf{M}_{\scriptscriptstyle{\text{env}}}^{\scriptscriptstyle{\Gamma,\mu}}$ solving the associated martingale problem. Hence $\mathbf{M}_{\scriptscriptstyle{\text{env}}}^{\scriptscriptstyle{\Gamma,\mu}}$ solves (\ref{envpro}) weakly and describes the motion of the environment seen from the tagged particle.

Finally, as the operator realizing $L_{\scriptscriptstyle{\scriptscriptstyle{\text{coup}}}}$ (see (\ref{igencoup})) we consider
\begin{multline*}
L^{\scriptscriptstyle{\mathbb{R}^d\times\Gamma,\hat{\mu}}}_{\scriptscriptstyle{\text{coup}}}\mathfrak{F}(\xi,\gamma)=L^{\scriptscriptstyle{\Gamma,\mu}}_{\scriptscriptstyle{\text{env}}}F(\gamma)\,f(\xi)-2\,\Big(\nabgam F(\gamma),\nabla f(\xi)\Big)_{\mathbb{R}^d}\\
+\sum_{x\in\gamma}\Big(\nabla\phi(x),\nabla f(\xi)\Big)_{\mathbb{R}^d}+\Delta f(\xi)\,F(\gamma)
\quad\mbox{for }d\xi\otimes\mu\mbox{-a.e.~}(\xi,\gamma)\in\mathbb{R}^d\times\Gamma,
\end{multline*}
where $\mathfrak{F}\in C_0^\infty(\mathbb{R}^d)\otimes \mathcal{F}C_b^{\infty}(C^{\infty}_0(\mathbb{R}^d),\Gamma)$ with
$\mathfrak{F}(x,\gamma)=f(x)\,F(\gamma)$ for $f\in C_0^\infty(\mathbb{R}^d)$, $F\in \mathcal{F}C_b^{\infty}(C^{\infty}_0(\mathbb{R}^d),\Gamma)$ and $\mu\in\mathcal{G}_{\scriptscriptstyle{ibp}}^{\scriptscriptstyle{gc}}(\Phi_{\scriptscriptstyle{\phi}},z\exp(-\phi))$, $0<z<\infty$. The associated symmetric bilinear form is given by
\begin{multline*}
\mathcal{E}^{\scriptscriptstyle{\mathbb{R}^d\times\Gamma,\hat{\mu}}}_{\scriptscriptstyle{\text{coup}}}(\mathfrak{F},\mathfrak{G})
=\int_{\mathbb{R}^d\times\Gamma}\Big((\nabgam -\nabla)\mathfrak{F}(\xi,\gamma),(\nabgam -\nabla)\mathfrak{G}(\xi,\gamma)\Big)_{\mathbb{R}^d}
\\+\Big(\nabla^\Gamma \mathfrak{F}(\xi,\gamma),\nabla^\Gamma \mathfrak{G}(\xi,\gamma)\Big)_{T_\gamma\Gamma}\,d\xi\,d{\mu}(\gamma),\quad\\
\mathfrak{F},\mathfrak{G}\in C_0^\infty(\mathbb{R}^d)\otimes \mathcal{F}C_b^{\infty}(C^{\infty}_0(\mathbb{R}^d),\Gamma).
\end{multline*}
Applying a strategy as used for tackling $\mathcal{E}^{\scriptscriptstyle{\Gamma,\mu}}_{\scriptscriptstyle{\text{gsdad}}}$ and $\mathcal{E}^{\scriptscriptstyle{\Gamma,\mu}}_{\scriptscriptstyle{\text{env}}}$ we have that also $(\mathcal{E}^{\scriptscriptstyle{\mathbb{R}^d\times\Gamma,\hat{\mu}}}_{\scriptscriptstyle{\text{coup}}},D(\mathcal{E}^{\scriptscriptstyle{\mathbb{R}^d\times\Gamma,\hat{\mu}}}_{\scriptscriptstyle{\text{coup}}}))$ is a conservative, local, quasi-regular Dirichlet form, where $\hat{\mu}:=d\xi\otimes\mu$. Therefore, there exists a conservative diffusion process $\mathbf{M}_{\scriptscriptstyle{\text{coup}}}^{\scriptscriptstyle{\mathbb{R}^d\times\Gamma,\hat{\mu}}}$ taking values in $\mathbb{R}^d\times\Gamma$ for $d\ge 2$ (for $d=1$ the process exists only in the larger space $\mathbb{R}^d\times\ddot{\Gamma}$, where $\ddot{\Gamma}$ is the configuration space of multiple configurations) solving the martingale problem associated to (\ref{tppro}) and (\ref{envpro}). Thus $\mathbf{M}_{\scriptscriptstyle{\text{coup}}}^{\scriptscriptstyle{\mathbb{R}^d\times\Gamma,\hat{\mu}}}$ realizes the coupling of the motion of the tagged particle and the motion of the environment seen from the tagged particle. Then we obtain the tagged particle process by a projection of the process $\mathbf{M}_{\scriptscriptstyle{\text{coup}}}^{\scriptscriptstyle{\mathbb{R}^d\times\Gamma,\hat{\mu}}}$ to its first component. Note that the resulting process in general is no longer a Markov process.  

The progress achieved in this paper may be summarized by the following list of core results:
\begin{itemize}
\item
We prove an integration by parts formula for $\nabla^\Gamma$ with respect to grand canonical Gibbs measures $\mu$ fulfilling a Ruelle bound and having $\sigma=z\,\exp(-\phi)\,dx,~0<z<\infty$, as intensity measure, see Theorems \ref{thmintbyparts}. 

\item
We provide a rigorous explicit representation of the generator $L^{\scriptscriptstyle{\mathbb{R}^d\times\Gamma,\hat{\mu}}}_{\scriptscriptstyle{\text{coup}}}$ of the coupled process for functions in $C_0^\infty(\mathbb{R}^d)\otimes \mathcal{F}C_b^{\infty}(C^{\infty}_0(\mathbb{R}^d),\Gamma)$, see Theorem \ref{thmexcoup}.

\item
We prove quasi-regularity for $(\mathcal{E}^{\scriptscriptstyle{\Gamma,\mu}}_{\scriptscriptstyle{\text{env}}},D(\mathcal{E}^{\scriptscriptstyle{\Gamma,\mu}}_{\scriptscriptstyle{\text{env}}}))$, the Dirichlet form corresponding to the environment process, see Lemma \ref{lemE22}.

\item We show the existence of the tagged particle process with interaction potential rigorously by using Dirichlet form techniques, see Theorem \ref{thmexprocoup} and Remark \ref{remextppro}.

\item
The process we construct is conservative and the unique solution to the martingale problem corresponding to the Friedrichs' extension of $(L^{\scriptscriptstyle{\mathbb{R}^d\times\Gamma,\hat{\mu}}}_{\scriptscriptstyle{\text{coup}}},C_0^\infty(\mathbb{R}^d)\otimes\mathcal{F}C_b^{\infty}(C^{\infty}_0(\mathbb{R}^d),\Gamma))$, see Theorem \ref{thmexprocoup}.

\item
Our results give the first mathematically rigorous and complete construction of the tagged particle process in continuum with interaction potential.
\end{itemize}
Here we would like to stress that all the above results hold for a very general class of interaction potentials. We only have to assume that the interaction potential is super stable (SS), integrable (I), lower regular (LR), differentiable and $L^q$ (D$\text{L}^\text{q}$), $q>d\ge 1$, and locally summable (LS). Hence we can treat interaction potentials which might have a singularity at the origin, non-trivial negative part and infinite range as e.g.~the Lennard--Jones potential.


\section{Configuration spaces and Gibbs measures}\label{defcangm}

\subsection{Configuration space and Poisson measure}

Let ${\mathbb R}^{d},~d \in \mathbb{N}$, be equipped with the norm 
$|\cdot|_{{\mathbb R}^{d}}$
given by the Euclidean scalar product $(\cdot, \cdot)_{{\mathbb
R}^{d}}$. By ${\mathcal B}({\mathbb R}^{d})$ we denote the
corresponding Borel $\sigma$-algebra. ${\mathcal O}_c({\mathbb
R}^{d})$ denotes the system of all open sets in ${\mathbb
R}^{d}$, which have compact closure and $\mathcal{B}_c(\mathbb{R}^d)$ the sets from $\mathcal{B}(\mathbb{R}^d)$ having compact closure. The Lebesgue measure on the
measurable space $({\mathbb R}^{d}, {\mathcal B}({\mathbb R}^{d}))$ we
denote by $dx$.

The \emph{configuration space}\index{configuration space} $\Gamma$ over
${\mathbb R}^{d}$ is defined by
\begin{eqnarray*}
\Gamma := \left\{ \gamma \subset {\mathbb R}^{d} \big| \, |\gamma \cap K| < \infty \,
\,\, \mbox{for any compact} \, K \subset {\mathbb R}^{d} \right\}.
\end{eqnarray*}
Here $|A|$ denotes the cardinality of a set $A$.
Via the identification of $\gamma \in \Gamma$ with
$\sum_{x \in \gamma} \varepsilon_{x} \in {\mathcal M}_p({\mathbb R}^{d})$,
where $\varepsilon_{x}$ denotes the Dirac measure in $x \in
{\mathbb R}^{d}$, $\Gamma$ can be considered as a subset of the
set ${\mathcal M}_p({\mathbb R}^{d})$ of all positive, integer-valued Radon measures\index{Radon measure!positive, integer-valued}
on ${\mathbb R}^{d}$. Hence $\Gamma$ can be topologized by the
vague topology\index{vague topology}, i.e., the topology generated by maps
\begin{eqnarray}\label{liftmap}
\gamma \mapsto \, \langle f,\gamma  \rangle \, := \int_{{\mathbb
R}^{d}} f(x) \,d\gamma(x) = \sum_{x \in \gamma} f(x),
\end{eqnarray}
where $f \in C_{0}({\mathbb R}^{d})$,
the set of continuous functions on ${\mathbb R}^{d}$ with compact support.
We denote by ${\mathcal B}({\Gamma})$ the corresponding Borel $\sigma$-algebra. For a fixed intensity measure $\sigma$ on $(\mathbb{R}^d,\mathcal{B}(\mathbb{R}^d))$ we denote by $\pi_{\scriptscriptstyle{\sigma}}$ the Poisson measure on $(\Gamma,\mathcal{B}(\Gamma))$ with intensity measure $\sigma$. Fore more details, see e.g.~\cite{AKR98a}, \cite{Ka83} and \cite{KMM78}.    

\subsection{Grand canonical and canonical Gibbs measures}\label{gcgm}

Let $\phi$ be a symmetric pair potential\index{pair potential}, i.e., a measurable function
$\phi: {\mathbb R}^d \to {\mathbb R} \cup \{+ \infty \}$ such that
$\phi(x) = \phi(-x)\in\mathbb{R}$ for $x\in\mathbb{R}^d\setminus\{0\}$. 
Any pair potential $\phi$ defines a potential $\Phi_{\scriptscriptstyle{\phi}}$ as follows. We set
\begin{align*}
\Phi_{\scriptscriptstyle{\phi}}(\gamma):=0\mbox{ for }|\gamma|\not=2\quad\text{and}\quad\Phi_{\scriptscriptstyle{\phi}}(\gamma)=\phi(x-y)\mbox{ for }\gamma=\{x,y\}\subset\mathbb{R}^d.
\end{align*} 

For a given pair potential $\phi$ we define the \emph{potential energy}\index{potential energy} $E:\Gamma\to\mathbb{R}\cup\{+\infty\}$ by
\begin{align*}
\gamma\mapsto E(\gamma):=
\left\{\begin{array}{ll}
  \sum_{\{x,y\}\subset\gamma}\phi(x-y), & \text{if }\sum_{\{x,y\}\subset\gamma}|\phi(x-y)|<\infty\\
  +\infty, & \text{otherwise}
  \end{array}\right.,
\end{align*}
where the sum over the empty set is defined to be zero.

The \emph{interaction energy}\index{interaction energy} between to configurations $\gamma$ and $\eta$ from $\Gamma$ is defined by
\begin{align*}
W(\gamma\,|\,\eta):=
\left\{\begin{array}{ll}
	\sum_{x\in\gamma,y\in\eta}\phi(x-y), & \text{if }\sum_{x\in\gamma,y\in\eta}|\phi(x-y)|<\infty\\
	+\infty, & \text{otherwise}
\end{array}\right.
\end{align*}
(typically we have $\gamma\cap\eta=\varnothing$).

In our terminology for any $\Lambda\in\mathcal{O}_c(\mathbb{R}^d)$ the \emph{conditional energy }\index{conditional energy}$E_{\scriptscriptstyle{\Lambda}}:\Gamma\to\mathbb{R}\cup\{+\infty\}$ is given by
\begin{align*}
\gamma\mapsto E_{\scriptscriptstyle{\Lambda}}(\gamma):=E(\gamma_\Lambda)+W(\gamma_\Lambda\,|\,\gamma_{\Lambda^c}).
\end{align*}
To introduce \emph{grand canonical Gibbs measures} on $(\Gamma,\mathcal{B}(\Gamma))$ we need the notion of a \emph{Gibbsian specification}\index{Gibbsian specification}. For any $\Lambda\in\mathcal{O}_c(\mathbb{R}^d)$ the specification $\Pi^{\scriptscriptstyle{\sigma}}_{\scriptscriptstyle{\Lambda}}$ is defined for
any $\gamma\in\Gamma$, $\Delta\in\mathcal{B}(\Gamma)$, by (see e.g.~\cite{Pr76})
\begin{align*}
\Pi^{\scriptscriptstyle{\sigma}}_{\scriptscriptstyle{\Lambda}}(\gamma,\Delta):=1_{\scriptscriptstyle{\left\{Z^{\scriptscriptstyle{\sigma}}_{\scriptscriptstyle{\Lambda}}<\infty\right\}}}(\gamma)\left(Z^{\scriptscriptstyle{\sigma}}_{\scriptscriptstyle{\Lambda}}(\gamma)\right)^{-1}\!\!\int_{\Gamma}
\!\!1_{\Delta}(\gamma_{{\Lambda}^c}\!\cup\!\gamma'_{{\Lambda}})\exp\left(\!-\!E_{\scriptscriptstyle{\Lambda}}(\gamma_{{\Lambda}^c}\!\cup\!\gamma'_{{\Lambda}})\right)d\pi_{\scriptscriptstyle{\sigma}}(\gamma'),
\end{align*}
where
\begin{align*}
Z^{\scriptscriptstyle{\sigma}}_{\scriptscriptstyle{\Lambda}}(\gamma):=\int_{\Gamma}\exp\left(\!-\!E_{\scriptscriptstyle{\Lambda}}(\gamma_{{\Lambda}^c}\!\cup\!\gamma'_{{\Lambda}})\right)d\pi_{\scriptscriptstyle{\sigma}}(\gamma')
\end{align*}
and $1_{\scriptscriptstyle{\left\{Z^{\scriptscriptstyle{\sigma}}_{\scriptscriptstyle{\Lambda}}<\infty\right\}}}$ denotes the indicator function of the set $\{\gamma\in \Gamma\,|\,Z^{\scriptscriptstyle{\sigma}}_{\scriptscriptstyle{\Lambda}}(\gamma)<\infty\}$.
A probability measure $\mu$ on $(\Gamma,\mathcal{B}(\Gamma))$, we write $\mu\in\mathcal{M}^1(\Gamma,\mathcal{B}(\Gamma))$, is called a grand canonical Gibbs measure corresponding to the potential $\Phi_{\scriptscriptstyle{\phi}}$ and the intensity measure $\sigma$ if it satisfies the \emph{Dobrushin-Lanford-Ruelle-equation (DLR)}:
\begin{align*}
\mu\,\Pi^{\scriptscriptstyle{\sigma}}_{\scriptscriptstyle{\Lambda}}=\mu\quad\mbox{for
all }\Lambda\in\mathcal{O}_c(\mathbb{R}^d). 
\end{align*}
For $\Lambda\in\mathcal{O}_c(\mathbb{R}^d)$ define for $\gamma\in\Gamma,~\Delta\in\mathcal{B}(\Gamma)$
\begin{align*}
\hat{\Pi}^{\scriptscriptstyle{\sigma}}_{\scriptscriptstyle{\Lambda}}(\gamma,\Delta):=\left\{
\begin{array}{ll}
\frac{\Pi^{\scriptscriptstyle{\sigma}}_{\scriptscriptstyle{\Lambda}}(\gamma,\Delta\cap\{\eta\in\Gamma\,|\,\eta(\Lambda)=\gamma(\Lambda)\})}{\Pi^{\scriptscriptstyle{\sigma}}_{\scriptscriptstyle{\Lambda}}(\gamma,\{\eta\in\Gamma\,|\,\eta(\Lambda)=\gamma(\Lambda)\})}, & \mbox{if }\Pi^{\scriptscriptstyle{\sigma}}_{\scriptscriptstyle{\Lambda}}(\gamma,\{\eta\in\Gamma\,|\,\eta(\Lambda)=\gamma(\Lambda)\})>0\\
0, & \mbox{otherwise}
\end{array}
\right..
\end{align*}
A probability measure $\mu$ on $(\Gamma,\mathcal{B}(\Gamma))$ is called a {\it canonical Gibbs measure}\index{Gibbs measure!canonical} to the potential $\Phi_{\scriptscriptstyle{\phi}}$ and the intensity $\sigma$ if
\begin{align*}
\mu\,\hat{\Pi}^{\scriptscriptstyle{\sigma}}_{\scriptscriptstyle{\Lambda}}=\mu\quad\mbox{for
all }\Lambda\in\mathcal{O}_c(\mathbb{R}^d).
\end{align*}
In the sequel we assume that the intensity measure $\sigma$ is absolutely continuous with respect to the Lebesgue measure with a bounded, non-negative density $\varrho$ and an activity parameter $0<z<\infty$,~i.e.,~$\frac{d\sigma}{dx}=z\varrho$, $0<z<\infty$. We then denote by $\mathcal{G}^{\scriptscriptstyle{gc}}(\Phi_{\scriptscriptstyle{\phi}},z\varrho)$, $0<z<\infty$, the set of corresponding grand canonical Gibbs measures and by $\mathcal{G}^{\scriptscriptstyle{c}}(\Phi_{\scriptscriptstyle{\phi}},\varrho)$, the set of corresponding canonical Gibbs measures. Due to \cite[Prop.~2.1]{Pr79} we have for given potential $\Phi_{\scriptscriptstyle{\phi}}$ and a bounded, non-negative density function $\varrho$ that
\begin{align}\label{inclu}
\mathcal{G}^{\scriptscriptstyle{\text{gc}}}(\Phi_{\scriptscriptstyle{\phi}},z\varrho)\subset\mathcal{G}^{\scriptscriptstyle{\text{c}}}(\Phi_{\scriptscriptstyle{\phi}},\varrho),~0<z<\infty.
\end{align}

\subsection{$K$-transform and correlation measures}\label{ss23}

Next, we recall the definition of correlation functions using
the concept of the $K$-transform, see \cite{KK99a} for a detailed study. Denote by $\Gamma_0$ the space of finite configurations over $\mathbb{R}^d$:
\begin{align*}
\Gamma_0 := \bigsqcup_{n=0}^\infty \Gamma_{0}^{(n)},\quad\Gamma^{(n)}_{0}:=\{\eta\subset \mathbb{R}^d\,|\,|\eta|=n\},\quad\Gamma^{(0)}_{0}:=\{\varnothing\}.
\end{align*}
Let $\widetilde{\mathbb{R}^{d\times n}}:=\{(x_1,\ldots,x_n)\in \mathbb{R}^{d\times n}\,|\,x_k\not=x_j,~j\not=k\}$ and let $S^n$ denote the group of all permutations of $\{1,\dots,n\}$. Through the natural bijection $\widetilde{\mathbb{R}^{d\times n}}/S^n \longleftrightarrow \Gamma_{0}^{(n)}$
one defines a topology on $\Gamma^{(n)}_{0}$. Let ${\mathcal B}(\Gamma^{(n)}_{0})$ denote the Borel $\sigma$-algebra on $\Gamma^{(n)}_{0}$. We equip $\Gamma_0$ with the topology $\mathcal{O}(\Gamma_0)$ of disjoint union. The Borel $\sigma$-algebra we denote by $\mathcal{B}(\Gamma_0)$. A ${\mathcal B}(\Gamma_0)$-measurable function $G \colon \Gamma_0 \to
{\mathbb R}$, $G\in L^0(\Gamma_0)$ for short, is said to have bounded support if there exist
$\Lambda \in {\mathcal O}_c({\mathbb R}^{d})$ and $N \in {\mathbb
N}$ such that $\mbox{supp}(G) \subset \bigsqcup_{n=0}^N
\Gamma_{0, \Lambda}^{(n)}$, where $\Gamma_{0, \Lambda}^{(n)}:=\{\eta\subset\Lambda\,|\,|\eta|=n\}$.
For any $\gamma \in {\Gamma}$ let $\sum_{\eta \Subset \gamma}$
denote the summation over all $\eta \subset \gamma$ such that
$|\eta| < \infty$. For a function
$G: \Gamma_0 \to {\mathbb R}$, the $K$-transform\index{$K$-transform} of $G$ is
defined by
\begin{eqnarray}\label{eq9}
(KG)(\gamma):= \sum_{\eta \Subset \gamma} G(\eta)
\end{eqnarray}
for each $\gamma \in \Gamma$ such that at least one of the series
$\sum_{\eta \Subset \gamma} G^+(\eta)$ or $\sum_{\eta \Subset
\gamma} G^-(\eta)$ converges, where $G^{+} := \max \{ 0, G\}$ and
$G^{-} := -\min \{ 0, G\}$.  The convolution $\star$ is defined by
\begin{multline}\label{propconv}
\star:L^0(\Gamma_0)\times L^0(\Gamma_0)\to L^0(\Gamma_0)\\
(G_1,G_2)\mapsto (G_1\star G_2)(\eta):=\sum_{(\xi_1,\xi_2,\xi_3)\in{\mathcal{P}^3_{\varnothing}}(\eta)}G_1(\xi_1\cup\xi_2)G_2(\xi_2\cup\xi_3),
\end{multline}
where $\mathcal{P}^3_{\varnothing}(\eta)$ denotes the set of all partitions $(\xi_1,\xi_2,\xi_3)$ of $\eta\in\Gamma_0$ in $3$ parts, i.e., all triples $\xi_i\subset\eta,~\xi_i\cap\xi_j=\varnothing$ if $i\not=j$, and $\xi_1\cup\xi_2\cup\xi_3=\eta$. We say $G\in L^0_{\text{ls}}(\Gamma_0)$ iff $G\in L^0(\Gamma_0)$ and there exists $\Lambda\in\mathcal{O}_c(\mathbb{R}^d)$ such that $G|_{\Gamma_0\setminus\Gamma_\Lambda}=0$. I.e., functions in $L^0_{\text{ls}}(\Gamma_0)$ are locally supported. Let $G_1, G_2\in L^0_{\text{ls}}(\Gamma_0)$. Then due to \cite[Prop.~3.11]{KK99a}.
\begin{align*}
K(G_1\star G_2)=KG_1\,KG_2.
\end{align*}
Let $\mu$ be a probability measure on $(\Gamma,{\mathcal B}(\Gamma))$. The
correlation measure\index{correlation measure} corresponding to $\mu$ is defined by
\begin{eqnarray*}
\rho_\mu(A) := \int_{\Gamma}(K1_A)(\gamma) \,d\mu(\gamma), \qquad A
\in {\mathcal B}(\Gamma_0).
\end{eqnarray*}
$\rho_\mu$ is a measure on $(\Gamma_0, {\mathcal B} (\Gamma_0))$
(see \cite{KK99a} for details, in particular, measurability issues).

Let $G \in L^1(\Gamma_0,\rho_\mu)$,
then $ \| KG \|_{L^1(\Gamma,\mu)} \le \| K|G| \|_{L^1(\Gamma,\mu)}
= \| G \|_{L^1(\Gamma_0,\rho_\mu)}$,
hence $KG \in L^1(\Gamma,\mu)$ and $KG(\gamma)$
is for $\mu$-a.e.~$\gamma \in \Gamma$ absolutely convergent.
Moreover, then obviously
\begin{eqnarray}\label{eq302}
\int_{\Gamma_0} G(\eta) \,d\rho_\mu(\eta) =
\int_{\Gamma}(KG)(\gamma) \,d\mu(\gamma),
\end{eqnarray}
see \cite{KK99a}, \cite{Len75a}, \cite{Len75b}.

For any $\mu\in\mathcal{G}^{\scriptscriptstyle{gc}}(\Phi_{\scriptscriptstyle{\phi}},z\,\varrho),~0<z<\infty$, the correlation measure $\rho_\mu$ is absolutely continuous with respect to the
Lebesgue-Poisson measure $\lambda_{\sigma}$, see e.g.~\cite[Rem.~4.4]{KK99a} and the references therein.
Its Radon-Nikodym derivative
\begin{eqnarray*}
\rho_{\mu}(\eta) := \frac{d\rho_\mu}{d\lambda_{\sigma}}(\eta), \qquad \eta \in
\Gamma_0,
\end{eqnarray*}
with respect to~$\lambda_{\sigma}$ we denote by the same symbol and the functions
\begin{eqnarray*}
\rho_{\mu}^{(n)}(x_1,\dots,x_n) := \rho_{\mu}(\{x_1,\dots,x_n\}),
\quad x_1,\dots,x_n \in {\mathbb R}^d, \,\, x_i \neq x_j
\,\, \mbox {if} \,\, i \neq j,
\end{eqnarray*}
are called the $n$-th order correlation functions\index{correlation function} of the measure
$\mu$.

We put the following restriction on the correlation measures under consideration.

\begin{description}
\item[(RB)]
We say that a correlation measure $\rho_\mu:\mathcal{B}(\Gamma_0)\to (0,\infty)$ corresponding to a measure $\mu$ on $(\Gamma,\mathcal{B}(\Gamma))$ fulfills the \emph{Ruelle-bound}\index{Ruelle-bound}, if for some $C_R\in (0,\infty)$
\begin{align*}
\rho_\mu(\gamma)\le(C_R)^{|\gamma|},\quad\mbox{for }\lambda_\sigma\mbox{-a.a. }\gamma\in\Gamma_0.
\end{align*}
\end{description}
Denote by $\mathcal{G}_{\scriptscriptstyle{Rb}}^{\scriptscriptstyle{gc}}(\Phi_{\scriptscriptstyle{\phi}},z\varrho)$, $0<z<\infty$, the set of all grand canonical Gibbs measures from $\mathcal{G}^{\scriptscriptstyle{gc}}(\Phi_{\scriptscriptstyle{\phi}},z\varrho)$, $0<z<\infty$, which fulfill (RB).

\subsection{Conditions on the interactions}\label{subcondpot}

For every $r = (r_1, \ldots, r_d) \in {\mathbb Z}^d$ we define a cube
\begin{eqnarray*}
Q_r = \Big{\{} x \in {\mathbb R}^d \, \Big{|} \, r_i - 1/2 \le x_i < r_i + 1/2 \Big{\}}.
\end{eqnarray*}
These cubes form a partition of ${\mathbb R}^d$. For any $\gamma \in \Gamma$ we set
$\gamma_r := \gamma_{Q_r}, \, r \in {\mathbb Z}^d$. Additionally, we introduce for
$n \in {\mathbb N}$ the cube $\Theta_n$ with side length $2n -1$ centered at the origin
in ${\mathbb R}^d$.

\begin{description}
\item[(SS)] ({\itshape Superstability})\index{Superstability~(SS)} There exist $0<A<\infty,~0\le B <\infty$ such that, if $\gamma =
\gamma_{\Theta_n}$ for some $n \in {\mathbb N}$, then
\begin{eqnarray*}
E_{\Theta_n}(\gamma) \, \ge \, \sum_{r \in {\mathbb Z}^d}
\Big{(} A |\gamma_r|^2 - B |\gamma_r| \Big{)}.
\end{eqnarray*}
\end{description}
(SS) obviously implies:
\begin{description}
\item[(S)] ({\itshape Stability})\index{Stability~(S)}
For any $\Lambda \in {\mathcal O}_c({\mathbb R}^{d})$ and for all
$\gamma \in \Gamma$ we have
\begin{eqnarray*}
E_{\Lambda}(\gamma) \, \ge \, -B |\gamma_{\Lambda}|,\quad 0\le B <\infty.
\end{eqnarray*}
\end{description}
As a consequence of (S), in turn, we have that $\phi$ is bounded from below. 
We also need
\begin{description}
\item[(I)] ({\itshape Integrability})\index{Integrability~(I)} We have:
\begin{eqnarray*}
 \int_{{\mathbb R}^{d}} | \exp(- \phi(x)) - 1 | \,dx < \infty.
\end{eqnarray*}
\end{description}

\begin{description}
\item [(LR)] ({\itshape Lower Regularity})\index{Lower Regularity~(LR)} There exists a decreasing positive function $a: {\mathbb N}
\to (0,\infty)$ such that
\begin{eqnarray*}
\sum_{r \in {\mathbb Z}^d} a(\| r \|_{\max}) < \infty
\end{eqnarray*}
and for any $\Lambda^{\prime}, \Lambda^{\prime \prime}$ which are finite unions
of cubes of the form $Q_r$ and disjoint,
\begin{eqnarray*}
W(\gamma^{\prime} \mid \gamma^{\prime \prime}) \ge
- \sum_{r^{\prime}, r^{\prime \prime} \in {\mathbb Z}^d} a(\|
r^{\prime} -  r^{\prime \prime} \|_{\max}) \,
|\gamma^{\prime}_{r^\prime}| \,
|\gamma^{\prime \prime}_{r^{\prime \prime}}|,
\end{eqnarray*}
provided $\gamma^{\prime}
= \gamma^{\prime}_{\Lambda^{\prime}}, \, \gamma^{\prime \prime}
= \gamma^{\prime \prime}_{\Lambda^{\prime \prime}}$.\\
Here and below $\|x\|_{\max}:=\max_{1\le i\le d}|x_i|,\quad x=(x_1,\ldots,x_d)\in\mathbb{R}^d$.
\end{description}

\begin{remark}
Using an argumentation as in \cite[Prop.~2.17]{KK01} the notion of \emph{Lower Regularity} (LR) given here implies the one defined in \cite[Sect.~2.5]{KK01}. Note that we are dealing with an intensity measure $\sigma=z\varrho\,dx$, $0<z<\infty$, where $\varrho$ is a bounded, non-negative density.
\end{remark}


\begin{description}
\item[(D{\bf $\text{L}^\text{q}$})]({\itshape Differentiability and $L^q$})\index{Differentiability and $L^p$~(D$\mathbf{L^p}$)} The function $\exp(-\phi)$ is weakly
differentiable on $\mathbb{R}^d$, $\phi$ is weakly differentiable on
$\mathbb{R}^d\setminus\{0\}$. The gradient $\nabla\phi$, considered
as a $dx$-a.e.~defined function on $\mathbb{R}^d$, satisfies
\begin{align*}
\nabla\phi \in L^1(\mathbb{R}^d,\exp(-\phi)\,dx)\cap L^q(\mathbb{R}^d,\exp(-\phi)dx),\quad 1\le q<\infty.
\end{align*}
\end{description}

\begin{remark}
Note that for many typical potentials in Statistical Physics we have $\phi\in C^\infty(\mathbb{R}^d\setminus\{0\})$. For such \lq\lq regular outside the origin\rq\rq potentials condition (D$\text{L}^\text{2}$) nevertheless does not exclude a singularity at the point $0\in\mathbb{R}^d$.
\end{remark} 

Let $(\Omega_n)_{n\in\mathbb{N}}$ be a partition of $\mathbb{R}^d$ in $\mathcal{B}_c(\mathbb{R}^d)$, i.e.~$\Omega_n\cap\Omega_m=\varnothing$ for $m\not=n$, $n,m\in\mathbb{N}$, and $\bigsqcup_{n=1}^\infty\Omega_n=\mathbb{R}^d$. We set
\begin{align*}
\Gamma_{\scriptscriptstyle{\text{fd}}}\big((\Omega_n)_{n\in\mathbb{N}}\big):=\bigcup_{M\in\mathbb{N}}\bigcap_{n\in\mathbb{N}}\big\{\gamma\in\Gamma\,\big|\,|\gamma_{\Omega_n}|\le M\sigma(\Omega_n)\big\}.
\end{align*}
$\Gamma_{\scriptscriptstyle{\text{fd}}}\big((\Omega_n)_{n\in\mathbb{N}}\big)$ is called the \emph{set of configurations of finite density}. Furthermore, we set $\Lambda_n:=B_n$, $n\in\mathbb{N}$, where $B_r$, $r\in(0,\infty)$, denotes the open ball with radius $r$ around the origin with respect to the euclidean norm on $\mathbb{R}^d$.

\begin{description}
\item[(LS)]({\itshape Local Summability})
Let $\Omega_1:=\Lambda_1$ and $\Omega_n:=\Lambda_n\setminus\Lambda_{n-1}$ for $n\ge 2$. Assume that $\sigma(\Omega_n)\ge \kappa\,(n+1)$, for some $\kappa\in(0,\infty)$ and all $n\in\mathbb{N}$. For all $\Lambda$ in $\mathcal{O}_c(\mathbb{R}^d)$ and all $\gamma\in \Gamma_{\scriptscriptstyle{\text{fd}}}\big((\Omega_n)_{n\in\mathbb{N}}\big)$ we have
\begin{align*}
\lim_{n\to\infty}\sum_{y\in\gamma_{\scriptscriptstyle{\Lambda_n\setminus\Lambda}}}\nabla\phi(\cdot-y)\mbox{ exists in $L^1_{\text{loc}}(\Lambda,\sigma)$}.
\end{align*}
\end{description}

\begin{remark}
$ $
\begin{enumerate}
\item[(i)]
Note that in the case $\varrho=\exp(-\phi)$ the assumption $\sigma(\Omega_n)\ge \kappa\,(n+1)$ for some $\kappa\in (0,\infty)$ and all $n\in\mathbb{N}$, is fulfilled, whenever the potential $\phi$ is bounded outside of a set $\Lambda\in\mathcal{B}_c(\mathbb{R}^d)$. 
\item[(ii)]
In the case $\sigma(\Omega_n)\ge \kappa\,(n+1)$ for some $\kappa\in (0,\infty)$ and all $n\in\mathbb{N}$, one has for $\mu\in\mathcal{G}_{\scriptscriptstyle{Rb}}^{\scriptscriptstyle{gc}}(\Phi_{\scriptscriptstyle{\phi}},z\varrho),~0<z<\infty$, that $\mu(\Gamma_{\scriptscriptstyle{\text{fd}}}\big((\Omega_n)_{n\in\mathbb{N}}\big))=1$, due to \cite[Theo.~5.4]{KK01}. In this case the grand canonical Gibbs measure $\mu$ is called \emph{tempered}. 
\item[(iii)]
Condition (LS) seems to be more complicated to check. In \cite[Exam.~4.1]{AKR98b}, however, it is shown that the assumption
\begin{align*}
\|\nabla\phi(x)\|_{\max}\le\frac{C}{\| x\|_{\max}^\alpha},\quad \|x\|_{\max} \ge R,
\end{align*}
for some $0<R,C<\infty,\alpha>d+1$, together with (D$\text{L}^\text{2}$) implies (LS). In our setting the proof is exactly the same as given there. 
\end{enumerate}
\end{remark}

A concrete  example fulfilling our assumptions is the \emph{Lennard--Jones potential}~(see Figure \ref{figpot} below).

\begin{figure}[h]
\begin{center}
\includegraphics[height=5cm]{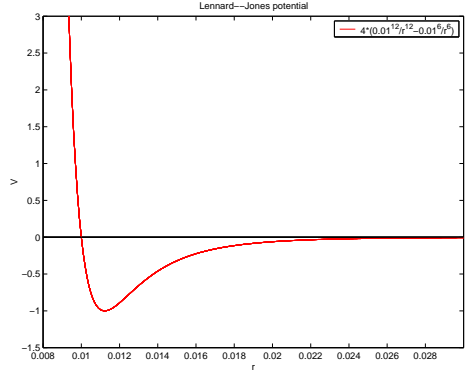}
\end{center}
\caption{\label{figpot}A typical example: The $(6,12)$-Lennard--Jones potential, i.e.~$\phi(x)=0.04\left(\frac{1}{|x|^{12}}-\frac{1}{|x|^{6}}\right),\quad x\in\mathbb{R}^d\setminus\{0\}$.\index{Lennard--Jones potential}}
\end{figure}

\subsection{Analysis and geometry on configuration spaces}\label{secgeometry}

On $\Gamma$ we define the set of smooth cylinder functions
\begin{align*}
\mathcal{F}C_b^\infty(C_0^\infty(\mathbb{R}^d),\Gamma):=\!\!\Big\{g(\langle
f_1,\gamma\rangle,\ldots,\langle
f_N,\gamma\rangle)\,\!\!\left|\!\,N\in\mathbb{N},~g\in
C_b^\infty(\mathbb{R}^N), f_1,\ldots,
f_N\in C_0^\infty(\mathbb{R}^d)\right.\!\!\Big\}.
\end{align*}
Clearly, $\mathcal{F}C_b^\infty( C_0^\infty(\mathbb{R}^d),\Gamma)$ is dense in
$L^2(\Gamma,\mathcal{B}(\Gamma),\pi_{{\sigma}})$.

Let $V_0(\mathbb{R}^d)$ denote the set of smooth
vector fields on $\mathbb{R}^d$. For $v\in V_0(\mathbb{R}^d)$ the \emph{directional derivatives on }$\Gamma$ for any $F=g_{\scriptscriptstyle{F}}(\langle f_1,\cdot\rangle,\ldots,\langle
f_N,\cdot\rangle)\in\mathcal{F}C_b^\infty( C_0^\infty(\mathbb{R}^d),\Gamma)$ are given by
\begin{multline}\label{equtbundle}
\nabla_v^\Gamma F(\gamma)=\sum_{i=1}^N\partial_i
g_{\scriptscriptstyle{F}}\left(\langle
f_1,\gamma\rangle,\ldots,\langle
f_N,\gamma\rangle\right)\langle\nabla_v f_i,\gamma\rangle\\
=\int_{\mathbb{R}^d}\left(\sum_{i=1}^N\partial_i
g_{\scriptscriptstyle{F}}\left(\langle
f_1,\gamma\rangle,\ldots,\langle
f_N,\gamma\rangle\right)\nabla f_i,v\right)_{\mathbb{R}^d}\,d\gamma
=\left(\nabla^\Gamma
F(\gamma),v\right)_{\scriptscriptstyle{L^2(\mathbb{R}^d\to\mathbb{R}^d,\gamma)}},
\end{multline}
with $\nabla_v f_i:=(\nabla
f_i,v)_{\scriptscriptstyle{\mathbb{R}^d}},~1\le i\le N$, $\gamma\in\Gamma$. Here $\nabla$ denotes the gradient on $\mathbb{R}^d$, $\partial_i$
the directional derivative with respect to~the $i$-th coordinate for $1\le i\le
N$ and $L^2(\mathbb{R}^d\to\mathbb{R}^d,\gamma)$ the space of
$\gamma$-square integrable vector fields on $\mathbb{R}^d$. 

Next we define a gradient\index{gradient} for functions in $\mathcal{F}C_b^\infty( C_0^\infty(\mathbb{R}^d),\Gamma)$ which corresponds to the directional
derivatives in (\ref{equtbundle}). So let
$F=g_{\scriptscriptstyle{F}}(\langle f_1,\cdot\rangle,\ldots,\langle
f_N,\cdot\rangle)\in\mathcal{F}C_b^\infty( C_0^\infty(\mathbb{R}^d),\Gamma),~v\in V_0(\mathbb{R}^d)$ and
$\gamma\in\Gamma$. The \emph{gradient} $\nabla^\Gamma$ of $F\in\mathcal{F}C_b^\infty( C_0^\infty(\mathbb{R}^d),\Gamma)$ at $\gamma\in\Gamma$ is defined by
\begin{align}\label{defgrad}
\Gamma\ni\gamma\mapsto\nabla^\Gamma F(\gamma):=\sum_{i=1}^N\partial_i
g_{\scriptscriptstyle{F}}\left(\langle
f_1,\gamma\rangle,\ldots,\langle f_N,\gamma\rangle\right)\nabla
f_i\in L^2(\mathbb{R}^d\to\mathbb{R}^d,\gamma).
\end{align}
Equation (\ref{equtbundle}) immediately leads to the appropriate
\emph{tangent space to}\index{tangent space} $\Gamma$, namely
\begin{align*}
T_{{\gamma}}\Gamma:=L^2(\mathbb{R}^d\to\mathbb{R}^d,\gamma),\quad\gamma\in\Gamma,
\end{align*}
equipped with the usual $L^2$-inner product. Note that $\nabla^\Gamma F$ is independent of the representation of $F$ in (\ref{defgrad}) and $\nabla^\Gamma F(\gamma)\in T_{{\gamma}}\Gamma$. The corresponding \emph{tangent bundle}\index{tangent bundle} is 
\begin{align*}
T\Gamma=\bigcup_{\gamma\in\Gamma}T_\gamma\Gamma.
\end{align*}

\emph{Finitely based vector fields on}\index{finitely based!vector field}
$(\Gamma,T\Gamma)$ can be defined as follows:
\begin{align*}
\Gamma\ni\gamma\mapsto\sum_{i=1}^N F_i(\gamma)v_i\in
V_0(\mathbb{R}^d),
\end{align*}
where $F_1,\ldots, F_N\in\mathcal{F}C_b^\infty( C_0^\infty(\mathbb{R}^d),\Gamma),~v_1,\ldots,v_N\in
V_0(\mathbb{R}^d)$. Let
$\mathcal{FV}C_b^\infty( C_0^\infty(\mathbb{R}^d),\Gamma)$
be the set of all such maps. Note that $\nabla^\Gamma F\in \mathcal{FV}C_b^\infty( C_0^\infty(\mathbb{R}^d),\Gamma)$ for all $F\in\mathcal{F}C_b^\infty( C_0^\infty(\mathbb{R}^d),\Gamma)$ and that each $v\in V_0(\mathbb{R}^d)$
is identified with the vector field $\gamma\mapsto v$ in $T\Gamma$
which is constant modulo taking $\gamma$-classes. For details we refer to \cite{AKR98a}, \cite{AKR98b}.

\section{An Integration by parts formula}\label{secintbp1}
In this section our aim is to prove an integration by parts formula for functions in\\ $\mathcal{F}C_b^\infty(C_0^\infty(\mathbb{R}^d),\Gamma)$ with respect to $\mu\in \mathcal{G}_{\scriptscriptstyle{\text{Rb}}}^{\scriptscriptstyle{\text{gc}}}(\Phi_\phi,z\exp(-\phi)),~0<z<\infty$, where $\phi$ fulfills (SS), (I) and (LR). Note that $\mathcal{G}_{\scriptscriptstyle{\text{Rb}}}^{\scriptscriptstyle{\text{gc}}}(\Phi_\phi,z\exp(-\phi)),~0<z<\infty$, is not empty see e.g.~\cite{CoKu09}. The following considerations are along the lines of \cite[Chap.~4.3]{AKR98b}.
We start with a technical lemma. 
\begin{lemma}\label{lemconv}
Let $\phi$ be a pair potential satisfying conditions (SS), (I), (LR) and (D$\text{L}^\text{2}$). For any vector field $v\in V_0(\mathbb{R}^d)$ we consider the function
\begin{align*}
\Gamma\ni\gamma\mapsto L_{v,k}^\phi(\gamma):=\left(\sum_{x\in\gamma_{\Lambda_k}}\!\!\!\!\big(\nabla\phi(x),v(x)\big)_{\mathbb{R}^d}\right)\!\!+\!\!\left(-\!\sum_{\{x,y\}\subset\gamma_{\Lambda_k}}\!\!\!\!\!\!\big(\nabla\phi(x-y),v(x)-v(y)\big)_{\mathbb{R}^d}\right)\in\mathbb{R}.
\end{align*}
Then for any $\mu\in\mathcal{G}^{\scriptscriptstyle{gc}}_{\scriptscriptstyle{\text{Rb}}}(\Phi_{\scriptscriptstyle{\phi}},z\exp(-\phi)),~0<z<\infty$, and all $v\in V_0(\mathbb{R}^d)$ we have that
\begin{align*}
L_v^{\phi,\mu}:=\lim_{k\to\infty}L_{v,k}^\phi
\end{align*}
exists in $L^2(\Gamma,\mu)$. Here $\Lambda_k,~k\in\mathbb{N}$, is defined as in Section \ref{defcangm}.
\end{lemma}

\begin{proof}
Let us at first consider the second summand. We set
\begin{align*}
\varphi^{\scriptscriptstyle{(2)}}_k(x,y):=\left|\big(1_{\Lambda_k}(x)1_{\Lambda_k}(y)\nabla\phi(x-y),v(x)-v(y)\big)_{\mathbb{R}^d}\right|
\end{align*}
and define
\begin{align*}
V^{\scriptscriptstyle{(2)}}_k(\gamma):=\left\{
\begin{array}{cc}
\varphi_k^{\scriptscriptstyle{(2)}}(x,y),\quad &\mbox{if }\gamma=\{x,y\}\in\Gamma^{\scriptscriptstyle(2)}_{0}\\
0,&\mbox{otherwise}	
\end{array}\right..
\end{align*}
Then by using (\ref{propconv}) and (\ref{eq302}), 
\begin{multline*}
\int_{\Gamma}\Big|-\!\!\!\!\sum_{\{x,y\}\subset\gamma_{\Lambda_k}}\!\!\!\!\!\!\big(\nabla\phi(x-y),v(x)-v(y)\big)_{\mathbb{R}^d}\Big|^2\,d\mu(\gamma)\\
\le \int_{\Gamma}\Big(\sum_{\{x,y\}\subset\gamma}\big|\big(1_{\Lambda_k}(x)1_{\Lambda_k}(y)\nabla\phi(x-y),v(x)-v(y)\big)_{\mathbb{R}^d}\big|\Big)^2\,d\mu(\gamma)\\
=\int_\Gamma \Big(\big(KV^{\scriptscriptstyle{(2)}}_k\big)(\gamma)\Big)^2\,d\mu(\gamma)
=\int_\Gamma \big(K(V^{\scriptscriptstyle{(2)}}_k\star V^{\scriptscriptstyle{(2)}}_k)\big)(\gamma)\,d\mu(\gamma)\\
=\int_{\Gamma_0}(V^{\scriptscriptstyle{(2)}}_k\star V^{\scriptscriptstyle{(2)}}_k)(\eta)\,d\rho_\mu(\eta)
=\int_{\Gamma_0}\sum_{(\xi_1,\xi_2,\xi_3)\in\mathcal{P}^3_{\varnothing}(\eta)}V^{\scriptscriptstyle{(2)}}_k(\xi_1\cup\xi_2)\,V^{\scriptscriptstyle{(2)}}_{k}(\xi_2\cup \xi_3)\,d\rho_\mu(\eta)\\
=\frac{1}{4!}\int_{\mathbb{R}^{4d}}\varphi^{\scriptscriptstyle{(2)}}_k(x_1,x_2)\varphi^{\scriptscriptstyle{(2)}}_k(x_3,x_4)\,\rho_\mu^{\scriptscriptstyle{(4)}}(x_1,x_2,x_3,x_4)\,d\sigma^{\otimes 4}\\
+\frac{1}{3!}\int_{\mathbb{R}^{3d}}\varphi^{\scriptscriptstyle{(2)}}_k(x_1,x_2)\varphi^{\scriptscriptstyle{(2)}}_k(x_2,x_3)\,\rho_\mu^{\scriptscriptstyle{(3)}}(x_1,x_2,x_3)\,d\sigma^{\otimes 3}\\
+\frac{1}{2!}\int_{\mathbb{R}^{2d}}\varphi^{\scriptscriptstyle{(2)}}_k(x_1,x_2)^2\,\rho_\mu^{\scriptscriptstyle{(2)}}(x_1,x_2)\,d\sigma^{\otimes 2}\\
\le C^{\scriptscriptstyle{(1)}}\int_{\mathbb{R}^{4d}}\varphi^{\scriptscriptstyle{(2)}}_k(x_1,x_2)\varphi^{\scriptscriptstyle{(2)}}_k(x_3,x_4)\,\rho_\mu^{\scriptscriptstyle{(4)}}(x_1,x_2,x_3,x_4)\,dx_1\ldots dx_4\\
+C^{\scriptscriptstyle{(2)}}\int_{\mathbb{R}^{3d}}\varphi^{\scriptscriptstyle{(2)}}_k(x_1,x_2)\varphi^{\scriptscriptstyle{(2)}}_k(x_2,x_3)\,\rho_\mu^{\scriptscriptstyle{(3)}}(x_1,x_2,x_3)\,dx_1\ldots dx_3\\
+C^{\scriptscriptstyle{(3)}}\int_{\mathbb{R}^{2d}}\varphi^{\scriptscriptstyle{(2)}}_k(x_1,x_2)^2\,\rho_\mu^{\scriptscriptstyle{(2)}}(x_1,x_2)\,dx_1\,dx_2,
\end{multline*}
where in the last step we have used the boundedness of the density function $\varrho=\exp(-\phi)$ and $0<C^{\scriptscriptstyle{(m)}}<\infty,~m\in\{1,2,3\}$.
The Mayer-Montroll equation\index{Mayer-Montroll equation} for correlation measures, see e.g.~\cite{KK99a}, together with (RB) and (I), gives
\begin{align*}
|\rho_\mu(x_1,\ldots x_p)|\le R_p\exp\left(-\sum_{i<j}\phi(x_j-x_i)\right),\quad 0<R_p<\infty,
\end{align*}
for all $p\in\mathbb{N},~x_1,\ldots,x_p\in\mathbb{R}^d$. From this point on we can proceed as in the proof of \cite[Lem.~4.1]{AKR98b}.\\
For the first summand we set
\begin{align*}
\varphi^{\scriptscriptstyle{(1)}}_k(x):=\left|\big(1_{\Lambda_k}(x)\nabla\phi(x),v(x)\big)_{\mathbb{R}^d}\right|
\end{align*}
and define correspondingly
\begin{align*}
V^{\scriptscriptstyle{(1)}}_k(\gamma):=\left\{
\begin{array}{cc}
\varphi_k^{\scriptscriptstyle{(1)}}(x),\quad &\mbox{if }\gamma=\{x\}\in\Gamma^{\scriptscriptstyle(1)}_{0}\\
0,&\mbox{otherwise}	
\end{array}\right..
\end{align*}
Thus we obtain by using (\ref{propconv}) and (\ref{eq302}),
\begin{multline*}
\int_{\Gamma}\Big|-\sum_{x\in\gamma_{\Lambda_k}}\big(\nabla\phi(x),v(x)\big)_{\mathbb{R}^d}\Big|^2\,d\mu(\gamma)\\\le\int_{\Gamma_0}\sum_{(\xi_1,\xi_2,\xi_3)\in\mathcal{P}^3_{\varnothing}(\eta)}V^{\scriptscriptstyle{(1)}}_k(\xi_1\cup \xi_2)V^{\scriptscriptstyle{(1)}}_{k}(\xi_2\cup \xi_3)\,d\rho_\mu(\eta)\\
=\frac{z^2}{2}\int_{\mathbb{R}^{2d}}\varphi^{\scriptscriptstyle{(1)}}_k(x_1)\varphi^{\scriptscriptstyle{(1)}}_k(x_2)\,\rho_\mu^{\scriptscriptstyle{(2)}}(x_1,x_2)\exp(-\phi(x_1))\exp(-\phi(x_2))\,dx_1\,dx_2\\
+z\int_{\mathbb{R}^d}\varphi^{\scriptscriptstyle{(1)}}_k(x_1)^2\,\rho_\mu^{\scriptscriptstyle{(1)}}(x_1)\exp(-\phi(x_1))\,dx_1\\
\le C^{\scriptscriptstyle{(4)}}\left(\int_{\Lambda_k^2}\|\nabla\phi(x)\|_{\max}\exp(-\phi(x))\,\| v(x)\|_{\max}\,\,dx\right)^2\\
+C^{\scriptscriptstyle{(5)}}\int_{\Lambda_k}\|\nabla\phi(x)\|_{\max}^2\exp(-\phi(x))\,\| v(x)\|_{\max}^2\,dx\\
\le C^{\scriptscriptstyle{(6)}}(v)\|\nabla\phi\|^2_{L^1(\Lambda_k,\exp(-\phi)dx)}+C^{\scriptscriptstyle{(7)}}(v)\|\nabla\phi\|^2_{L^2(\Lambda_k,\exp(-\phi)dx)}<\infty
\end{multline*}
with $C^{\scriptscriptstyle{(4)}}, C^{\scriptscriptstyle{(5)}}, C^{\scriptscriptstyle{(6)}}(v),C^{\scriptscriptstyle{(7)}}(v)\in(0,\infty),$ due to condition (D$\text{L}^\text{2}$) and $v\in V_0(\mathbb{R}^d)$. Finally since $\Lambda_k\uparrow\mathbb{R}^d$ as $k\to\infty$, it easily follows that $\left(L^{\phi}_{v,k}\right)_{k\in\mathbb{N}}$ is a Cauchy sequence in $L^2(\Gamma,\mu)$ and since this space is complete, the limit exists.
\end{proof}

\begin{definition}\label{BL2}
Let $\phi$ be a pair potential satisfying conditions (SS), (I), (LR) and (D$\text{L}^\text{2}$). For $v\in V_0(\mathbb{R}^d)$ and $\mu\in\mathcal{G}^{\scriptscriptstyle{gc}}_{\scriptscriptstyle{\text{Rb}}}(\Phi_{\scriptscriptstyle{\phi}},z\exp(-\phi)),~0<z<\infty$, we define
\begin{align*}
B_v^{\phi,\mu}:=L^{\phi,\mu}_v+\big\langle\text{div~}v,\cdot\big\rangle\in L^2(\Gamma,\mu).
\end{align*}
\end{definition}

Note that $\big\langle\text{div~}v,\cdot\big\rangle\in L^2(\Gamma,\mu)$, since $\mu\in\mathcal{G}^{\scriptscriptstyle{gc}}_{\scriptscriptstyle{\text{Rb}}}(\Phi_{\scriptscriptstyle{\phi}},z\exp(-\phi)),~0<z<\infty$.
Now we are able to formulate an important result which is essential for our applications below.

\begin{theorem}\label{thmintbyparts}
Suppose that the pair potential $\phi$ satisfies (SS), (I), (LR), (D$\text{L}^\text{2}$) and (LS). 
Let $\mu\in\mathcal{G}^{\scriptscriptstyle{gc}}_{\scriptscriptstyle{\text{Rb}}}(\Phi_{\scriptscriptstyle{\phi}},z\exp(-\phi)),~0<z<\infty$. Then for $v\in V_0(\mathbb{R}^d)$ and $F,G\in\mathcal{F}C_b^\infty(C_0^\infty(\mathbb{R}^d),\Gamma)$ the following integration by parts formula holds:
\begin{align*}
\int_\Gamma\nabla^\Gamma_vF\,G\,d\mu(\gamma)=-\int_\Gamma F\,\nabla_v^\Gamma G\,d\mu(\gamma)-\int_\Gamma F\,G\,B_v^{\phi,\mu}\,d\mu(\gamma).
\end{align*}
\end{theorem}

\begin{proof}
Let $F=g_F(\langle f_1,\cdot\rangle,\ldots,\langle f_N,\cdot\rangle)\in\mathcal{F}C_b^\infty(C_0^\infty(\mathbb{R}^d),\Gamma)$, $v\in V_0(\mathbb{R}^d)$ and choose $\Lambda\in\mathcal{O}_c(\mathbb{R}^d)$ such that $\bigcup_{i=1}^N\text{supp~}f_i\cup\text{supp~}v\subset\Lambda$.  
Using (\ref{inclu}) we have $\mathcal{G}^{\scriptscriptstyle{gc}}_{\scriptscriptstyle{\text{Rb}}}(\Phi_{\scriptscriptstyle{\phi}},z\exp(-\phi))\subset\mathcal{G}^{\scriptscriptstyle{c}}(\Phi_{\scriptscriptstyle{\phi}},\exp(-\phi))$, $0<z<\infty$. Hence
\begin{multline}\label{numerator}
\int_\Gamma\nabla^\Gamma_v F\,d\mu(\gamma)=\int_\Gamma\hat{\Pi}_{\scriptscriptstyle{\Lambda}}^{\scriptscriptstyle{\sigma,\phi}}(\nabla_v^\Gamma F)\,d\mu(\gamma)\\
=\!\!\!\!\int_\Gamma\!\!\!\frac{\left(\int_{\Lambda^{\gamma(\Lambda)}}\!\!\!\nabla^\Gamma_vF(\gamma_{\Lambda^c}\!\!\cup\!\{x_1,\!\ldots\!, x_{\gamma(\Lambda)}\})\!\exp\!\!\left(\!-\!E_{\scriptscriptstyle{\Lambda}}(\gamma_{{\Lambda}^c}\!\cup\!\{x_1,\!\ldots\!,x_{\gamma(\Lambda)}\})\!\right)\!\!\varrho^{\otimes \gamma(\Lambda)}\!\!dx_1\!\ldots\! dx_{\gamma(\Lambda)}\right)}{\int_{\Lambda^{\gamma(\Lambda)}}\exp\left(\!-\!E_{\scriptscriptstyle{\Lambda}}(\gamma_{{\Lambda}^c}\!\cup\!\{x_1,\!\ldots\!,x_{\gamma(\Lambda)}\})\right)\,\varrho^{\otimes \gamma(\Lambda)}\,dx_1\!\ldots\! dx_{\gamma(\Lambda)}}\!d\mu(\gamma),
\end{multline}
where $\varrho^{\otimes\gamma(\Lambda)}:=\varrho(x_1)\cdot\ldots\cdot\varrho(x_{\gamma(\Lambda)})$, see Section \ref{defcangm}. 
Fix $n\in\mathbb{N}$ and $\gamma\in\{\eta\in\Gamma\,|\,\eta(\Lambda)=n\}\cap \Gamma_{\text{fd}}\big((\Omega_m)_{m\in\mathbb{N}}\big)$, where $(\Omega_m)_{m\in\mathbb{N}}$ corresponds to $(\Lambda_m)_{m\in\mathbb{N}}$ as in (LS). 
Using \cite[Coro.~5.8]{KK01} the numerator of the integrand in (\ref{numerator}) for such $\gamma$ equals to
\begin{multline*}
\lim_{m\to\infty}\!\int_{\Lambda^n}\!\!\!\!\nabla_v^\Gamma F\big(\gamma_{\Lambda^c}\cup\{x_1,\!\ldots\!,x_n\}\big)\exp\!\Big(-\!E_{\scriptscriptstyle{\Lambda}}\big(\gamma_{{\Lambda_m}\setminus\Lambda}\!\cup\!\{x_1,\!\ldots\!,x_n\}\big)\Big)\!\varrho^{\otimes n}\!(x_1,\ldots,x_n)dx_1\!\ldots\! dx_n\\
=\lim_{m\to\infty}\int_{\Lambda^n}\sum_{i=1}^N\Bigg(\partial_i g_F\!\bigg(\sum_{j=1}^n f_1(x_j),\ldots,\sum_{j=1}^n f_N(x_j)\bigg)\sum_{k=1}^n\nabla_v^{\mathbb{R}^d}f_i(x_k)\Bigg)\\
\times\,\exp\Big(\!-\!E_{\scriptscriptstyle{\Lambda}}\big(\gamma_{{\Lambda_m}\setminus\Lambda}\!\cup\!\{x_1,\ldots,x_n\}\big)\Big)\,\varrho^{\otimes n}(x_1,\ldots,x_n)\,dx_1\ldots dx_n\\
=\lim_{m\to\infty}\sum_{k=1}^n\int_{\Lambda^n}\Bigg(\nabla_{x_k}g_F\bigg(\sum_{j=1}^n f_1(x_j),\ldots,\sum_{j=1}^n f_N(x_j)\bigg),v(x_k)\Bigg)_{\mathbb{R}^d}\\
\times\,\exp\left(\!-\!E_{\scriptscriptstyle{\Lambda}}(\gamma_{{\Lambda_m}\setminus\Lambda}\!\cup\!\{x_1,\ldots,x_n\})\right)\varrho^{\otimes n}(x_1,\ldots,x_n)\,dx_1\ldots dx_n.
\end{multline*}
Here $\nabla_{x_k},~1\le k\le n$, denotes the gradient with respect to the $x_k$-th variable $(x_k\in\Lambda)$. Integrating by parts with respect to $x_k,~1\le k\le n$, we obtain
\begin{multline}\label{calcintbyp}
-\lim_{m\to\infty}\sum_{k=1}^n\int_{\Lambda^n}g_F\bigg(\sum_{j=1}^n f_1(x_j),\ldots, \sum_{j=1}^n f_N(x_j)\bigg)\\
\times\Bigg(\bigg(\!\nabla_{x_k}\!\Big(\sum_{1\le i<j}^n\!\!\phi(x_i-x_j)+\sum_{i=1}^n\sum_{y\in\gamma_{{\Lambda_m}\setminus\Lambda}}\!\!\phi(x_i-y)\Big),v(x_k)\!\!\bigg)_{\mathbb{R}^d}\\
\times\exp\Big(\!-\!E_{\scriptscriptstyle{\Lambda}}\big(\gamma_{{\Lambda_m}\setminus\Lambda}\!\cup\!\{x_1,\ldots x_n\}\big)\Big)\varrho^{\otimes n}(x_1,\ldots,x_n)\\
+\Big(\nabla_{x_k}\varrho^{\otimes n}(x_1,\ldots x_n),v(x_k)\Big)_{\mathbb{R}^d}\exp\Big(\!-\!E_{\scriptscriptstyle{\Lambda}}\big(\gamma_{{\Lambda_m}\setminus\Lambda}\!\cup\!\{x_1,\ldots, x_n\}\big)\Big)\\
+\exp\Big(\!-\!E_{\scriptscriptstyle{\Lambda}}\big(\gamma_{{\Lambda_m}\setminus\Lambda}\!\cup\!\{x_1,\ldots,x_n\}\big)\Big)\,\varrho^{\otimes n}(x_1,\ldots,x_n)\,\text{div~}v(x_k)\Bigg)\,dx_1\ldots dx_n\\
=-\lim_{m\to\infty}\int_{\Lambda^n}g_F\bigg(\sum_{j=1}^n f_1(x_j),\ldots, \sum_{j=1}^n f_N(x_j)\bigg)\\
\times\Bigg(\Bigg(\sum_{1\le i<j}^n\Big(\nabla\phi(x_i-x_j),v(x_i)-v(x_j)\Big)_{\mathbb{R}^d}+\sum_{i=1}^n\sum_{y\in\gamma_{{\Lambda_m}\setminus\Lambda}}\!\!\Big(\nabla\phi(x_i-y),v(x_i)\Big)_{\mathbb{R}^d}\Bigg)\\
\times\exp\Big(\!-\!E_{\scriptscriptstyle{\Lambda}}\big(\gamma_{{\Lambda_m}\setminus\Lambda}\!\cup\!\{x_1,\ldots,x_n\}\big)\Big)\varrho^{\otimes n}(x_1,\ldots,x_n)\\
+\Big(\sum_{i=1}^n\Big(-\nabla\phi(x_i),v(x_i)\Big)_{\mathbb{R}^d}\varrho^{\otimes n}(x_1,\ldots, x_n)\Big)\exp\Big(\!-\!E_{\scriptscriptstyle{\Lambda}}\big(\gamma_{{\Lambda_m}\setminus\Lambda}\!\cup\!\{x_1,\ldots,x_n\}\big)\Big)\\
+\exp\Big(\!-\!E_{\scriptscriptstyle{\Lambda}}\big(\gamma_{{\Lambda_m}\setminus\Lambda}\!\cup\!\{x_1,\ldots,x_n\}\big)\Big)\,\varrho^{\otimes n}(x_1,\ldots, x_n)\,\sum_{i=1}^n\text{div~}v(x_i)\Bigg)\,dx_1\ldots dx_n\\
=-\lim_{m\to\infty}\int_{\Lambda^n}\Bigg(F(\{x_1,\ldots,x_n\})\Bigg(\bigg(\sum_{1\le i<j}^n\Big(\nabla\phi(x_i-x_j),v(x_i)-v(x_j)\Big)_{\mathbb{R}^d}\\
+\!\!\sum_{i=1}^n\sum_{y\in\gamma_{{\Lambda_m}\setminus\Lambda}}\!\!\Big(\nabla\phi(x_i-y),v(x_i)\Big)_{\mathbb{R}^d}\bigg)-\bigg(\sum_{i=1}^n\Big(\nabla\phi(x_i),v(x_i)\Big)_{\mathbb{R}^d}\bigg)+\sum_{i=1}^n\text{div~}v(x_i)\Bigg)\\
\times\exp\Big(\!-\!E_{\scriptscriptstyle{\Lambda}}\big(\gamma_{{\Lambda_m}\setminus\Lambda}\!\cup\!\{x_1,\ldots,x_n\}\big)\Big)\,\varrho^{\otimes n}(x_1,\ldots,x_n)\,dx_1\ldots dx_n\\
=-\int_{\Lambda^n}\Bigg(F(\{x_1,\ldots,x_n\})\bigg(\sum_{1\le i<j}^n\Big(\nabla\phi(x_i-x_j),v(x_i)-v(x_j)\Big)_{\mathbb{R}^d}\\
+\!\!\sum_{i=1}^n\sum_{y\in\gamma_{{\Lambda}^c}}\!\!\Big(\nabla\phi(x_i-y),v(x_i)\Big)_{\mathbb{R}^d}-\sum_{i=1}^n\Big(\nabla\phi(x_i),v(x_i)\Big)_{\mathbb{R}^d}
+\sum_{i=1}^n\text{div~}v(x_i)\bigg)\Bigg)\\
\times\exp\Big(\!-\!E_{\scriptscriptstyle{\Lambda}}\big(\gamma_{{\Lambda}^c}\!\cup\!\{x_1,\ldots,x_n\}\big)\Big)\,\varrho^{\otimes n}(x_1,\ldots,x_n)\,dx_1\ldots dx_n.
\end{multline}
In the last step we have used (LS). Thus by (\ref{calcintbyp}), Lemma \ref{lemconv} and Definition \ref{BL2} we obtain that (\ref{numerator}) equals
\begin{align*}
\int_\Gamma\!\!\!\frac{\int_{\Lambda^n}\!\!FB_v^{\phi,\mu}(\gamma_{\Lambda^c}\!\!\cup\!\{x_1,\!\ldots\!, x_{n}\})\!\exp\!\!\left(\!-\!E_{\scriptscriptstyle{\Lambda}}(\gamma_{{\Lambda}^c}\!\cup\!\{x_1,\!\ldots\!,x_{n}\})\right)\!\!\varrho^{\otimes n} dx_1\!\ldots\! dx_{n}}{\int_{\Lambda^{n}}\exp\left(\!-\!E_{\scriptscriptstyle{\Lambda}}(\gamma_{{\Lambda}^c}\!\cup\!\{x_1,\!\ldots\!,x_{n}\})\right)\,\varrho^{\otimes n}\,dx_1\!\ldots\! dx_{n}}d\mu(\gamma).
\end{align*}
Therefore,
\begin{align}\label{equintbyparts}
\int_\Gamma\nabla_v^\Gamma F\,d\mu(\gamma)=-\int_\Gamma\hat{\Pi}_{\scriptscriptstyle{\Lambda}}^{\scriptscriptstyle{\sigma,\phi}}(FB_v^{\phi,\mu})\,d\mu(\gamma)=-\int_\Gamma FB_v^{\phi,\mu}\,d\mu(\gamma).
\end{align}
By the product rule for $\nabla_v^\Gamma$ on $\Gamma$ we obtain
\begin{align*}
\int_\Gamma\nabla_v^\Gamma(FG)\,d\mu(\gamma)=\int_\Gamma\nabla_v^\Gamma F\,G\,d\mu(\gamma)+\int_\Gamma F\,\nabla_v^\Gamma G\,d\mu(\gamma)
\end{align*}
and by (\ref{equintbyparts})
\begin{align*}
-\int_\Gamma FGB_v^{\phi,\mu}\,d\mu(\gamma)=\int_\Gamma\nabla_v^\Gamma F\,G\,d\mu(\gamma)+\int_\Gamma F\,\nabla_v^\Gamma G\,d\mu(\gamma).
\end{align*}
\end{proof}

For $V:=\sum_{i=1}^N F_iv_i\in\mathcal{FV}C_b^\infty(C^\infty_0(\mathbb{R}^d),\Gamma)$ we define
\begin{align}\label{equdiv}
\text{div}^{\scriptscriptstyle{\Gamma,{\mu}}} V:=\sum_{i=1}^N\left(\nabla^\Gamma_{v_i}
F_i+B^{\phi,\mu}_{{v_i}}F_i\right)
\end{align}
and for $F\in\mathcal{F}C_b^\infty(C^\infty_0(\mathbb{R}^d),\Gamma)$
\begin{align}\label{equL}
L^{\scriptscriptstyle{\Gamma,\mu}} F:=\text{div}^{\scriptscriptstyle{\Gamma,{\mu}}} \nabla^\Gamma F.
\end{align}

Note that $\nabla^\Gamma F\in\mathcal{FV}C_b^\infty(C^\infty_0(\mathbb{R}^d),\Gamma)$, since 
\begin{align*}
(\nabla^\Gamma F)(\gamma,x)=\sum_{i=1}^N\partial_i g_F(\langle f_1,\gamma\rangle,\ldots\langle f_N,\gamma\rangle)\nabla f_i(x),\quad\gamma\in\Gamma,\quad x\in\mathbb{R}^d.
\end{align*}

\begin{corollary}\label{corintbyparts}
Under the assumptions of Theorem \ref{thmintbyparts} we have for all\\ $F\in\mathcal{F}C_b^\infty(C^\infty_0(\mathbb{R}^d),\Gamma),~V\in\mathcal{FV}C_b^\infty(C^\infty_0(\mathbb{R}^d),\Gamma)$
\begin{align*}
\int_\Gamma\left(\nabla^\Gamma F,V\right)_{T_\gamma\Gamma}\,d\mu(\gamma)=-\int_\Gamma F\,\text{div}^{\scriptscriptstyle{\Gamma,\mu}} V\,d\mu(\gamma).
\end{align*}
\end{corollary}

\begin{proof}
Let $F\in\mathcal{F}C_b^\infty(C^\infty_0(\mathbb{R}^d),\Gamma),~V\in\mathcal{FV}C_b^\infty(C^\infty_0(\mathbb{R}^d),\Gamma)$. Hence $V(\gamma)=\sum_{i=1}^N G_i(\gamma)v_i$ for all $\gamma\in\Gamma$ and for some $G_i\in\mathcal{F}C_b^\infty(C^\infty_0(\mathbb{R}^d),\Gamma),~v_i\in V_0(\mathbb{R}^d),~1\le i\le N$. By (\ref{equtbundle})
\begin{align*}
\int_\Gamma\left(\nabla^\Gamma F,V\right)_{T_\gamma\Gamma}\,d\mu(\gamma)=\sum_{i=1}^N\int_\Gamma\nabla_{v_i}^\Gamma F\,G_i\,d\mu(\gamma).
\end{align*}
Now we apply Theorem \ref{thmintbyparts} and by (\ref{equdiv}) the statement follows.
\end{proof}

\section{Infinite Interacting Particle Systems}\label{secipp}

Suppose that the pair potential $\phi$ satisfies (SS), (I), (LR), (D$\text{L}^\text{2}$) and (LS).\\ 

\subsection{The gradient stochastic dynamics with additional drift}

We start with  
\begin{align*}
\mathcal{E}^{\scriptscriptstyle{\Gamma,\mu}}_{\scriptscriptstyle{gsdad}}(F,G):=\int_\Gamma\left(\nabla^\Gamma F(\gamma),\nabla^\Gamma G(\gamma)\right)_{T_\gamma\Gamma}\,d\mu(\gamma),\quad F,G\in\mathcal{F}C_b^\infty(C^\infty_0(\mathbb{R}^d),\Gamma).
\end{align*}
Our aim is to show that the closure $(\mathcal{E}^{\scriptscriptstyle{\Gamma,\mu}}_{\scriptscriptstyle{\text{gsdad}}},D(\mathcal{E}^{\scriptscriptstyle{\Gamma,\mu}}_{\scriptscriptstyle{\text{gsdad}}}))$ of $(\mathcal{E}^{\scriptscriptstyle{\Gamma,\mu}}_{\scriptscriptstyle{\text{gsdad}}},\mathcal{F}C_b^\infty(C^\infty_0(\mathbb{R}^d),\Gamma))$ is a conservative, local, quasi-regular Dirichlet form. By definition it is the classical gradient Dirichlet form on $L^2(\Gamma,\mu)$, but in our situation $\mu$ is a grand canonical Gibbs measure corresponding to the intensity measure $\sigma=z\,\exp(-\phi)\,dx,~0<z<\infty$. This is different to the classical situation, where grand canonical Gibbs measures $\mu$ corresponding to $\sigma=z\,dx,~0<z<\infty$, are considered, see e.g.~\cite{AKR98b}.   

\begin{remark}\label{remsymbimgsd}
$\left(\nabla^\Gamma F,\nabla^\Gamma G\right)_{T_\cdot\Gamma} \in L^1(\Gamma,\mu)$ because $\mu\in\mathcal{G}^{\scriptscriptstyle{gc}}_{\scriptscriptstyle{\text{Rb}}}(\Phi_{\scriptscriptstyle{\phi}},z\exp(-\phi)),~0<z<\infty$. Due to Theorem \ref{thmintbyparts} we have that $\nabla^\Gamma$ respects the $\mu$-classes $\mathcal{F}C_b^{\infty,\mu}(C^\infty_0(\mathbb{R}^d),\Gamma)$ determined by\\ $\mathcal{F}C_b^\infty(C^\infty_0(\mathbb{R}^d),\Gamma)$, i.e.,~$\nabla^\Gamma F=\nabla^\Gamma G~\mu$-a.e provided $F,G\in\mathcal{F}C_b^\infty(C^\infty_0(\mathbb{R}^d),\Gamma)$ satisfy $F=G~\mu$-a.e.. Furthermore, it is easy to check that the $\mu$-equivalence classes $\mathcal{FV}C_b^{\infty,\mu}(C^\infty_0(\mathbb{R}^d),\Gamma)$ determined by $\mathcal{FV}C_b^{\infty}(C^\infty_0(\mathbb{R}^d),\Gamma)$ are dense in $L^2(\mathbb{R}^d\to\mathbb{R}^d,\mu)$. 
Hence $\left(\mathcal{E}^{\scriptscriptstyle{\Gamma,\mu}}_{\scriptscriptstyle{\text{gsdad}}},\mathcal{F}C_b^{\infty,\mu}(C^\infty_0(\mathbb{R}^d),\Gamma)\right)$ is a densely defined positive definite symmetric bilinear form on $L^2(\Gamma,\mu)$. 
\end{remark}

The major part of the analysis (concerning closability) is already done by the derivation of the corresponding integration by parts formula in Section \ref{secintbp1}.  

\begin{corollary}\label{corgen}
Under the assumptions of Theorem \ref{thmintbyparts}. We have 
\begin{align*}
\mathcal{E}^{\scriptscriptstyle{\Gamma,\mu}}_{\scriptscriptstyle{\text{gsdad}}}(F,G)=\int_\Gamma\left(\nabla^\Gamma F(\gamma),\nabla^\Gamma G(\gamma)\right)_{T_\gamma\Gamma}\,d\mu(\gamma)=\int_\Gamma -L^{\scriptscriptstyle{\Gamma,\mu}}_{\scriptscriptstyle{\text{gsdad}}}F\,G\,d\mu
\end{align*}
for all $F,G\in\mathcal{F}C_b^\infty(C^\infty_0(\mathbb{R}^d),\Gamma)$. In particular,
\begin{multline*}
L^{\scriptscriptstyle{\Gamma,\mu}}_{\scriptscriptstyle{\text{gsdad}}} F(\gamma)=\sum_{i,j=1}^N\partial_i\partial_jg_{\scriptscriptstyle{F}}\Big(\langle f_1,\gamma\rangle,\ldots,\langle f_N,\gamma\rangle\Big)\left\langle\Big(\nabla f_i,\nabla f_j\Big)_{\mathbb{R}^d},\gamma\right\rangle\\
+\sum_{j=1}^N\partial_j g_{\scriptscriptstyle{F}}\Big(\langle f_1,\gamma\rangle,\ldots,\langle f_N,\gamma\rangle\Big)\bigg(\langle\Delta f_j,\gamma\rangle+\left\langle\Big(\nabla\phi,\nabla f_j\Big)_{\mathbb{R}^d},\gamma\right\rangle\\-\sum_{\{x,y\}\subset\gamma}\Big(\nabla\phi(x-y),\nabla f_j(x)-\nabla f_j(y)\Big)_{\mathbb{R}^d}\bigg)
\end{multline*}
$\mbox{for }\mu\mbox{-a.e.~}\gamma\in\Gamma\mbox{ and }F\in\mathcal{F}C_b^\infty(C^\infty_0(\mathbb{R}^d),\Gamma)$.
\end{corollary}

\begin{proof}
Apply Corollary \ref{corintbyparts} with $V:=\nabla^\Gamma G$. Then the first assertion follows by (\ref{equL}). The second we obtain by direct calculations using (\ref{equL}) and (\ref{defgrad}).
\end{proof}

In the sequel we denote by $\ddot{\Gamma}\subset\mathcal{M}_p(\mathbb{R}^d)$ the space of integer valued, positive Radon measures. Note that $\ddot{\Gamma}\supset\Gamma$, since
\begin{align*}
\Gamma=\left\{\gamma\in\ddot{\Gamma}\,\left|\,\max_{x\in\mathbb{R}^d}\gamma(\{x\})\le 1\right.\right\}. 
\end{align*}

\begin{remark}
Clearly, $\nabla^\Gamma$ extends to a linear operator on $D(\mathcal{E}^{\scriptscriptstyle{\Gamma,\mu}}_{\scriptscriptstyle{\text{gsdad}}})$. We denote these extension by the same symbol. Furthermore, note that since $\Gamma\subset\ddot{\Gamma}$ and $\mathcal{B}(\ddot{\Gamma})\cap\Gamma=\mathcal{B}(\Gamma)$ we can consider $\mu$ as a measure on $(\ddot{\Gamma},\mathcal{B}(\ddot{\Gamma}))$ and correspondingly $\left(\mathcal{E}^{\scriptscriptstyle{\Gamma,\mu}}_{\scriptscriptstyle{\text{gsdad}}},D(\mathcal{E}^{\scriptscriptstyle{\Gamma,\mu}}_{\scriptscriptstyle{\text{gsdad}}})\right)$ is a Dirichlet form on $L^2(\ddot{\Gamma},\mu)$. In particular, we have that $D(\mathcal{E}^{\scriptscriptstyle{\Gamma,\mu}}_{\scriptscriptstyle{\text{gsdad}}})$ is the closure of $\mathcal{F}C_b^{\infty,\mu}(C_0^{\infty}(\mathbb{R}^d),\ddot{\Gamma})$ with respect to the norm $\sqrt{{\mathcal{E}^{\scriptscriptstyle{\Gamma,\mu}}_{\scriptscriptstyle{\text{gsdad}}}}_1}$, where
\begin{align*}
{{\mathcal{E}^{\scriptscriptstyle{\Gamma,\mu}}_{\scriptscriptstyle{\text{gsdad}}}}_1}(F):={\mathcal{E}^{\scriptscriptstyle{\Gamma,\mu}}_{\scriptscriptstyle{\text{gsdad}}}(F,F)+(F,F)_{L^2(\ddot{\Gamma},\mu)}},\quad F\in D(\mathcal{E}^{\scriptscriptstyle{\Gamma,\mu}}_{\scriptscriptstyle{\text{gsdad}}}).
\end{align*}
The corresponding generator of the Dirichlet form can also be considered as linear operator on $L^2(\ddot{\Gamma},\mu)$.
\end{remark}

\begin{theorem}\label{thmform1}
Suppose that the pair potential $\phi$ satisfies (SS), (I), (LR), (D$\text{L}^\text{2}$) and (LS). 
Let $\mu\in\mathcal{G}^{\scriptscriptstyle{gc}}_{\scriptscriptstyle{\text{Rb}}}(\Phi_{\scriptscriptstyle{\phi}},z\exp(-\phi)),~0<z<\infty$. Then
\begin{enumerate}
\item[(i)]
$\left(\mathcal{E}^{\scriptscriptstyle{\Gamma,\mu}}_{\scriptscriptstyle{\text{gsdad}}},\mathcal{F}C_b^{\infty,\mu}(C^\infty_0(\mathbb{R}^d),\Gamma)\right)$ is closable on $L^2(\Gamma,\mu)$ and its closure $\left(\mathcal{E}^{\scriptscriptstyle{\Gamma,\mu}}_{\scriptscriptstyle{\text{gsdad}}},D(\mathcal{E}^{\scriptscriptstyle{\Gamma,\mu}}_{\scriptscriptstyle{\text{gsdad}}})\right)$ is a symmetric Dirichlet form which is conservative, i.e.,~$1\in D(\mathcal{E}^{\scriptscriptstyle{\Gamma,\mu}}_{\scriptscriptstyle{\text{gsdad}}}),~\mathcal{E}^{\scriptscriptstyle{\Gamma,\mu}}_{\scriptscriptstyle{\text{gsdad}}}(1,1)=0$. Its generator, denoted by $H^{\scriptscriptstyle{\Gamma,\mu}}_{\scriptscriptstyle{\text{gsdad}}}$, is the Friedrichs' extension of $-L^{\scriptscriptstyle{\Gamma,\mu}}_{\scriptscriptstyle{\text{gsdad}}}$.
\item[(ii)]
$\left(\mathcal{E}^{\scriptscriptstyle{\Gamma,\mu}}_{\scriptscriptstyle{\text{gsdad}}},D(\mathcal{E}^{\scriptscriptstyle{\Gamma,\mu}}_{\scriptscriptstyle{\text{gsdad}}})\right)$ is quasi-regular on $L^2(\ddot{\Gamma},\mu)$. 
\item[(iii)]
$\left(\mathcal{E}^{\scriptscriptstyle{\Gamma,\mu}}_{\scriptscriptstyle{\text{gsdad}}},D(\mathcal{E}^{\scriptscriptstyle{\Gamma,\mu}}_{\scriptscriptstyle{\text{gsdad}}})\right)$ is local, i.e., $\mathcal{E}^{\scriptscriptstyle{\Gamma,\mu}}_{\scriptscriptstyle{\text{gsdad}}}(F,G)=0$ provided $F,G\in D(\mathcal{E}^{\scriptscriptstyle{\Gamma,\mu}}_{\scriptscriptstyle{\text{gsdad}}})$ with\\ $\text{supp}(|F|\cdot\mu)\cap\text{supp}(|G|\cdot\mu)=\varnothing$.
\end{enumerate}
\end{theorem}

\begin{proof}
$ $
\begin{enumerate}
\item[(i)]
By Corollary \ref{corgen} we have closability and the last part of the assertion. The Dirichlet property immediately follows from the chain rule for $\nabla^\Gamma$ on $\mathcal{F}C_b^{\infty}(C^\infty_0(\mathbb{R}^d),\Gamma)$ and the conservativity is obvious. (We refer to \cite[Chap.~I and Chap.~II,~Sect.~2,3]{MaRo92} for the terminology and details.)
\item[(ii)]
This is a special case of \cite[Coro.~4.9]{MaRo00}.
\item[(iii)]
Since $\nabla_\Gamma$ satisfies the product rule on bounded functions in $D(\mathcal{E}^{\scriptscriptstyle{\Gamma,\mu}}_{\scriptscriptstyle{\text{gsdad}}})$ the proof is exactly the same as in \cite[Chap.~V, Exam.~1.12(ii)]{MaRo92}.
\end{enumerate}
\end{proof}

\begin{theorem}\label{thmexpromgsd}
Suppose the assumptions of Theorem \ref{thmform1}. Then
\begin{enumerate}
\item[(i)]
there exists a conservative diffusion process 
\begin{align*}
\mathbf{{M}}^{\scriptscriptstyle{\Gamma,\mu}}_{\scriptscriptstyle{\text{gsdad}}}=\left(\mathbf{{\Omega}},\mathbf{{F}}^{\scriptscriptstyle{\text{gsdad}}},(\mathbf{{F}}^{\scriptscriptstyle{\text{gsdad}}}_t)_{t\ge 0},(\mathbf{X}^{\scriptscriptstyle{\text{gsdad}}}_t)_{t\ge 0},(\mathbf{{P}}^{\scriptscriptstyle{\text{gsdad}}}_\gamma)_{\gamma\in\ddot{\Gamma}}\right)
\end{align*}
on $\ddot{\Gamma}$ which is properly associated with 
$\left(\mathcal{E}^{\scriptscriptstyle{\Gamma,\mu}}_{\scriptscriptstyle{\text{gsdad}}},D(\mathcal{E}^{\scriptscriptstyle{\Gamma,\mu}}_{\scriptscriptstyle{\text{gsdad}}})\right)$, i.e., for all ($\mu$-versions of) $F\in L^2(\ddot{\Gamma},\mu)$ and all $t>0$ the function
\begin{align*}
\gamma\mapsto p^{\scriptscriptstyle{\text{gsdad}}}_t F(\gamma):=\int_{{\mathbf{\Omega}}}F({\mathbf{X}}^{\scriptscriptstyle{\text{gsdad}}}_t)\,d{\mathbf{{P}}^{\scriptscriptstyle{\text{gsdad}}}}_\gamma,\quad\gamma\in\ddot{\Gamma},
\end{align*}
is an $\mathcal{E}^{\scriptscriptstyle{\Gamma,\mu}}_{\scriptscriptstyle{\text{gsdad}}}$-quasi-continuous version of $\exp(-t {H}^{\scriptscriptstyle{\Gamma,\mu}}_{\scriptscriptstyle{\text{gsdad}}})F$. $\mathbf{{M}}^{\scriptscriptstyle{\Gamma,\mu}}_{\scriptscriptstyle{\text{gsdad}}}$ is up to $\mu$-equivalence unique (cf.~\cite[Chap.~IV,~Sect.~6]{MaRo92}). In particular, $\mathbf{{M}}^{\scriptscriptstyle{\Gamma,\mu}}_{\scriptscriptstyle{\scriptscriptstyle{\text{gsdad}}}}$ is $\mu$-symmetric, i.e.,
\begin{align*}
\int_{\ddot{\Gamma}} G\,p^{\scriptscriptstyle{\text{gsdad}}}_t F\,d\mu(\gamma)=\int_{\ddot{\Gamma}}F\,p^{\scriptscriptstyle{\text{gsdad}}}_t G\,d\mu(\gamma)\quad\mbox{for all }F,G:\ddot{\Gamma}\to\mathbb{R_+},~\mathcal{B}(\ddot{\Gamma})\mbox{-measurable}
\end{align*}
and has $\mu$ as invariant measure.
\item[(ii)]
$\mathbf{{M}}^{\scriptscriptstyle{\Gamma,\mu}}_{\scriptscriptstyle{\text{gsdad}}}$ from (i) is the (up to $\mu$-equivalence,~cf.~\cite[Def.~6.3]{MaRo92}) unique diffusion process having $\mu$ as invariant measure and solving the martingale problem for\\ $\left(-{H}^{\scriptscriptstyle{\Gamma,\mu}}_{\scriptscriptstyle{\text{gsdad}}},D({H}^{\scriptscriptstyle{\Gamma,\mu}}_{\scriptscriptstyle{\text{gsdad}}})\right)$, i.e., for all $G\in D({H}^{\scriptscriptstyle{\Gamma,\mu}}_{\scriptscriptstyle{\text{gsdad}}})\supset\mathcal{F}C_b^\infty(C^\infty_0(\mathbb{R}^d),\Gamma)$
\begin{align*}
\widetilde{G}(\mathbf{X}^{\scriptscriptstyle{\text{gsdad}}}_t)-\widetilde{G}(\mathbf{X}^{\scriptscriptstyle{\text{gsdad}}}_0)+\int_0^t {H}^{\scriptscriptstyle{\Gamma,\mu}}_{\scriptscriptstyle{\text{gsdad}}} G(\mathbf{{X}}^{\scriptscriptstyle{\text{gsdad}}}_t)\,ds,\quad t\ge 0,
\end{align*}
is an $(\mathbf{{F}}^{\scriptscriptstyle{\text{gsdad}}}_t)_{t\ge 0}$-martingale under $\mathbf{{P}}^{\scriptscriptstyle{\text{gsdad}}}_\gamma$ (hence starting at $\gamma$) for $\mathcal{E}^{\scriptscriptstyle{\Gamma,\mu}}_{\scriptscriptstyle{\text{gsdad}}}$-q.a.~$\gamma\in\ddot{\Gamma}$. (Here $\widetilde{G}$ denotes a quasi-continuous version of $G$, cf.~\cite[Chap.~IV,~Prop.3.3]{MaRo92}.)
\end{enumerate}
\end{theorem}

\begin{proof}
$ $
\begin{enumerate}
\item[(i)]
By Theorem \ref{thmform1} the proof follows directly from \cite[Chap.~V,~Theo.~1.11]{MaRo92}.
\item[(ii)]
This follows immediately by \cite[Theo.~3.5]{MR1335494}.
\end{enumerate}
\end{proof}

\begin{remark}\label{remexceptmgsd}
\begin{enumerate}
\item[(i)]
For $d\ge 2$ an argumentation as in the proof of \cite[Prop.~1]{RS98} together with an argumentation as in the proof of \cite[Coro.~1]{RS98} gives us that under our assumptions the set $\ddot{\Gamma}\setminus\Gamma$ is $\mathcal{E}_{\scriptscriptstyle{\text{gsdad}}}^{\scriptscriptstyle{\Gamma,\mu}}$-exceptional. Therefore, the process $\mathbf{M}^{\scriptscriptstyle{\Gamma,\mu}}_{\scriptscriptstyle{\text{gsdad}}}$ from Theorem \ref{thmexpromgsd} lives on the smaller space $\Gamma$.
\item[(ii)]
We call the diffusion process $\mathbf{M}^{\scriptscriptstyle{\Gamma,\mu}}_{\scriptscriptstyle{\text{gsdad}}}$ from Theorem \ref{thmexpromgsd} \emph{gradient stochastic dynamics with additional drift}.
\end{enumerate}
\end{remark}

\subsection{The environment process}\label{secenv}\index{environment process}

The following statement is a special case of an integration by parts formula shown in \cite{CoKu09}, which holds for a non-empty subset $\mathcal{G}^{\scriptscriptstyle{gc}}_{\scriptscriptstyle{\text{ibp}}}(\Phi_{\scriptscriptstyle{\phi}},z\exp(-\phi))$ of $\mathcal{G}^{\scriptscriptstyle{gc}}_{\scriptscriptstyle{\text{Rb}}}(\Phi_{\scriptscriptstyle{\phi}},z\exp(-\phi))$, $0<z<\infty$. 
\begin{lemma}\label{corintbp2}
Suppose that the pair potential $\phi$ satisfies (SS), (I), (LR), (D$\text{L}^\text{2}$) and (LS). Let $\mu\in\mathcal{G}^{\scriptscriptstyle{gc}}_{\scriptscriptstyle{\text{ibp}}}(\Phi_{\scriptscriptstyle{\phi}},z\exp(-\phi)),~0<z<\infty$. Then for $F,G\in\mathcal{F}C_b^\infty(C^\infty_0(\mathbb{R}^d),\Gamma)$ we have
$\left(\nabla^\Gamma_\gamma F(\gamma),\nabla^\Gamma_\gamma G(\gamma)\right)_{\scriptscriptstyle{\mathbb{R}^d}}$
$\in L^1(\Gamma,\mu)$.
\pagebreak
Furthermore,
\begin{multline*}
\int_{\Gamma}\left(\nabla^\Gamma_\gamma F(\gamma),\nabla^\Gamma_\gamma G(\gamma)\right)_{\scriptscriptstyle{\mathbb{R}^d}}\,d\mu(\gamma)\\
=-\int_\Gamma\Bigg(\sum_{i,j=1}^N\partial_i\partial_j g_{\scriptscriptstyle{F}}\left(\langle f_1,\gamma\rangle,\ldots,\langle f_N,\gamma\rangle\right)\Big(\left\langle\nabla f_i,\gamma\right\rangle,\left\langle\nabla f_j,\gamma\right\rangle\Big)_{\mathbb{R}^d}\\
+\sum_{j=1}^N \partial_j g_{\scriptscriptstyle{F}}\left(\langle f_1,\gamma\rangle,\ldots,\langle f_N,\gamma\rangle\right)\bigg(\left\langle\Delta f_j,\gamma\right\rangle-\Big(\left\langle\nabla\phi,\gamma\right\rangle,\left\langle\nabla f_j,\gamma\right\rangle\Big)_{\mathbb{R}^d}\bigg)\Bigg)\\
\times g_{\scriptscriptstyle{G}}\left(\langle g_1,\gamma\rangle,\ldots,\langle g_M,\gamma\rangle\right)\,d\mu(\gamma).
\end{multline*}
\end{lemma}

Next we consider
\begin{align*}
\mathcal{E}^{\scriptscriptstyle{\Gamma,\mu}}_{\scriptscriptstyle{\text{env}}}(F,G)=\mathcal{E}_{\scriptscriptstyle{\text{gsdad}}}^{\scriptscriptstyle{\Gamma,\mu}}(F,G)+\int_{\Gamma}\left(\nabla^\Gamma_\gamma F(\gamma),\nabla^\Gamma_\gamma G(\gamma)\right)_{\scriptscriptstyle{\mathbb{R}^d}}\,d\mu(\gamma),\quad F,G\in\mathcal{F}C_b^\infty(C^\infty_0(\mathbb{R}^d),\Gamma).
\end{align*}

\begin{remark}\label{remsymbienv}
Using Remark \ref{remsymbimgsd} and Lemma \ref{corintbp2} we have that $\left(\mathcal{E}^{\scriptscriptstyle{\Gamma,\mu}}_{\scriptscriptstyle{\text{env}}},\mathcal{F}C_b^{\infty,\mu}(C^\infty_0(\mathbb{R}^d),\Gamma)\right)$ is a densely defined, positive definite, symmetric bilinear form on $L^2(\Gamma,\mu)$. 
\end{remark}

\begin{corollary}\label{corgenE2n}
Suppose that the pair potential $\phi$ satisfies (SS), (I), (LR), (D$\text{L}^\text{2}$) and (LS). Let $\mu\in\mathcal{G}^{\scriptscriptstyle{gc}}_{\scriptscriptstyle{\text{ibp}}}(\Phi_{\scriptscriptstyle{\phi}},z\exp(-\phi)),~0<z<\infty$. Then for all $F,G\in\mathcal{F}C_b^\infty(C^\infty_0(\mathbb{R}^d),\Gamma)$ we have 
\begin{align*}
\mathcal{E}^{\scriptscriptstyle{\Gamma,\mu}}_{\scriptscriptstyle{\text{env}}}(F,G)=\int_\Gamma\left(\nabla^\Gamma F(\gamma),\nabla^\Gamma G(\gamma)\right)_{T_\gamma\Gamma}+\left(\nabla^\Gamma_\gamma F(\gamma),\nabla^\Gamma_\gamma G(\gamma)\right)_{\scriptscriptstyle{\mathbb{R}^d}}\,d\mu(\gamma)=\int_\Gamma -L^{\scriptscriptstyle{\Gamma,\mu}}_{\scriptscriptstyle{\text{env}}}F\,G\,d\mu.
\end{align*}
In particular, 
\begin{multline*}
L^{\scriptscriptstyle{\Gamma,\mu}}_{\scriptscriptstyle{\text{env}}}F(\gamma)=L^{\scriptscriptstyle{\Gamma,\mu}}_{\scriptscriptstyle{\text{gsdad}}}F(\gamma)
+\sum_{i,j=1}^N\partial_i\partial_j g_{\scriptscriptstyle{F}}\left(\langle f_1,\gamma\rangle,\ldots,\langle f_N,\gamma\rangle\right)\big(\langle\nabla f_i,\gamma\rangle,\langle\nabla f_j,\gamma\rangle\big)_{\mathbb{R}^d}\\
+\sum_{j=1}^N \partial_j g_{\scriptscriptstyle{F}}\left(\langle f_1,\gamma\rangle,\ldots,\langle f_N,\gamma\rangle\right)\Big(\langle\Delta f_j,\gamma\rangle-\big(\langle\nabla\phi,\gamma\rangle,\langle\nabla f_j,\gamma\rangle\big)_{\mathbb{R}^d}\Big)\quad\mbox{for }\mu\mbox{-a.e.~}\gamma\in\Gamma.
\end{multline*}
\end{corollary}

\begin{proof}
Combining Corollary \ref{corgen} and Lemma \ref{corintbp2} the statement follows.
\end{proof}

\begin{lemma}\label{lemE21}
$\left(\mathcal{E}^{\scriptscriptstyle{\Gamma,\mu}}_{\scriptscriptstyle{\text{env}}},\mathcal{F}C_b^{\infty,\mu}(C^\infty_0(\mathbb{R}^d),\Gamma)\right)$ is closable on $L^2(\Gamma,\mu)$ and its closure
$\left(\mathcal{E}^{\scriptscriptstyle{\Gamma,\mu}}_{\scriptscriptstyle{\text{env}}},D(\mathcal{E}^{\scriptscriptstyle{\Gamma,\mu}}_{\scriptscriptstyle{\text{env}}})\right)$ is a symmetric Dirichlet form which is conservative, i.e.,~$1\in D(\mathcal{E}^{\scriptscriptstyle{\Gamma,\mu}}_{\scriptscriptstyle{\text{env}}})$ and $\mathcal{E}^{\scriptscriptstyle{\Gamma,\mu}}_{\scriptscriptstyle{\text{env}}}(1,1)=0$. Its generator, denoted by ${H}^{\scriptscriptstyle{\Gamma,\mu}}_{\scriptscriptstyle{\text{env}}}$, is the Friedrichs' extension of $-L^{\scriptscriptstyle{\Gamma,\mu}}_{\scriptscriptstyle{\text{env}}}$.
\end{lemma}

\begin{proof}
By Corollary \ref{corgenE2n} we have closability and the last part. The Dirichlet property immediately follows since $\nabla^\Gamma$ and $\nabla^\Gamma_\gamma$ fulfill the chain rule on $\mathcal{F}C_b^{\infty}(C^\infty_0(\mathbb{R}^d),\Gamma)$. Conservativity is obvious.
\end{proof}

\begin{remark}
Clearly, $\nabla^\Gamma$ and $\nabla^\Gamma_\gamma$ extend to linear operators on $D(\mathcal{E}^{\scriptscriptstyle{\Gamma,\mu}}_{\scriptscriptstyle{\text{env}}})$. We denote these extensions by the same symbols. Furthermore, note that since $\Gamma\subset\ddot{\Gamma}$ and $\mathcal{B}(\ddot{\Gamma})\cap\Gamma=\mathcal{B}(\Gamma)$ we can consider $\mu$ as a measure on $(\ddot{\Gamma},\mathcal{B}(\ddot{\Gamma}))$ and correspondingly $\left(\mathcal{E}^{\scriptscriptstyle{\Gamma,\mu}}_{\scriptscriptstyle{\text{env}}},D(\mathcal{E}^{\scriptscriptstyle{\Gamma,\mu}}_{\scriptscriptstyle{\text{env}}})\right)$ is a Dirichlet form on $L^2(\ddot{\Gamma},\mu)$. In particular, we have that $D(\mathcal{E}^{\scriptscriptstyle{\Gamma,\mu}}_{\scriptscriptstyle{\text{env}}})$ is the closure of $\mathcal{F}C_b^{\infty,\mu}(C_0^{\infty}(\mathbb{R}^d),\ddot{\Gamma})$ with respect to the norm $\sqrt{{\mathcal{E}^{\scriptscriptstyle{\Gamma,\mu}}_{\scriptscriptstyle{\text{env}}}}_1}$, where
\begin{align*}
{{\mathcal{E}^{\scriptscriptstyle{\Gamma,\mu}}_{\scriptscriptstyle{\text{env}}}}_1}(F):={\mathcal{E}^{\scriptscriptstyle{\Gamma,\mu}}_{\scriptscriptstyle{\text{env}}}(F,F)+(F,F)_{L^2(\ddot{\Gamma},\mu)}},\quad F\in D(\mathcal{E}^{\scriptscriptstyle{\Gamma,\mu}}_{\scriptscriptstyle{\text{env}}}).
\end{align*}
The corresponding generator of the Dirichlet form can also be considered as linear operator on $L^2(\ddot{\Gamma},\mu)$.
\end{remark}

\begin{lemma}\label{lemE22}
$\left(\mathcal{E}^{\scriptscriptstyle{\Gamma,\mu}}_{\scriptscriptstyle{\text{env}}},D(\mathcal{E}^{\scriptscriptstyle{\Gamma,\mu}}_{\scriptscriptstyle{\text{env}}})\right)$ is quasi-regular on $L^2(\ddot{\Gamma},\mu)$.
\end{lemma}

\begin{proof}
The Dirichlet form $(\mathcal{E}^{\scriptscriptstyle{\Gamma,\mu}}_{\scriptscriptstyle{\text{env}}},D(\mathcal{E}^{\scriptscriptstyle{\Gamma,\mu}}_{\scriptscriptstyle{\text{env}}})))$ is given by
\begin{align*}
\mathcal{E}^{\scriptscriptstyle{\Gamma,\mu}}_{\scriptscriptstyle{\text{env}}}(F,G):=\int_\Gamma S^\Gamma(F,G)\,d\mu,
\end{align*}
where
\begin{multline*}
S^\Gamma(F,G):=S^\Gamma_0(F,G)+\left(\nabla^\Gamma_\gamma F(\gamma),\nabla^\Gamma_\gamma G(\gamma)\right)_{\scriptscriptstyle{\mathbb{R}^d}}\quad\mbox{with}\\
S^\Gamma_0(F,G):=\left(\nabla^\Gamma F,\nabla^\Gamma G\right)_{T_\gamma\Gamma},\quad F,G\in D(\mathcal{E}^{\scriptscriptstyle{\Gamma,\mu}}_{\scriptscriptstyle{\text{env}}}).
\end{multline*}
To prove quasi-regularity analogously to \cite[Prop.~4.1]{MaRo00}, it suffices to show that there exists a
bounded, complete metric $\bar{\rho}$ on $\ddot{\Gamma}$ generating the vague topology such that $\bar{\rho}(\cdot,\gamma_0)\in D(\mathcal{E}^{\scriptscriptstyle{\Gamma,\mu}}_{\scriptscriptstyle{\text{env}}})$ for all $\gamma_0\in\ddot\Gamma$ and
\begin{align*}
S^\Gamma(\bar{\rho}(\cdot,\gamma_0),\bar{\rho}(\cdot,\gamma_0))\le\eta\quad\mu-\mbox{a.e.}
\end{align*}
for some $\eta\in L^1(\ddot{\Gamma},\mu)$ (independent of $\gamma_0$).
The proof below is a modification of \cite[Prop.~4.8]{MaRo00}. Hence we also use the notation proposed there.
Thus $(B_k)_{k\in\mathbb{N}}$ is an exhausting sequence, i.e.~$(B_k)_{k\in\mathbb{N}}$ is an increasing sequence of open sets such that $\bigcup_{k\in\mathbb{N}}B_k=\mathbb{R}^d$. Furthermore, since $B^{\frac{1}{2}}_k\subset B_{k+1}$ for all $k\in\mathbb{N}$, $(B_k)_{k\in\mathbb{N}}$ is a well-exhausting sequence in the sense of \cite{MaRo00} with $\delta_k=\frac{1}{2}$ for all $k\in\mathbb{N}$. Here $B^{\frac{1}{2}}_k:=B_{k+\frac{1}{2}}$. For each $k\in\mathbb{N}$ we define
\begin{align*}
g_k(x):=g_{B_k,\frac{1}{2}}(x):=\frac{2}{3}\left(\frac{1}{2}-\text{dist}(x,B_k)\wedge\frac{1}{2}\right),\quad x\in\mathbb{R}^d,
\end{align*}
and $\phi_k:=3g_k$. Furthermore, we set $S(f,g):=\big(\nabla f,\nabla g \big)_{\mathbb{R}^d}$ for $f,g\in W_0^{1,2}(\mathbb{R}^d)$, where $W_0^{1,2}(\mathbb{R}^d)$ denotes the Sobolev space of compactly supported, weakly differentiable functions in $L^2(\mathbb{R}^d,dx)$ with weak derivative again in $L^2(\mathbb{R}^d,dx)$. Due to \cite[Exam.~4.5.1]{MaRo00} we have that \cite[Cond.~(Q)]{MaRo00} holds with $S$ as given above. Moreover, due to \cite[Lemm.~4.10]{MaRo00}
\begin{align*}
\phi_k g_j\in W^{1,2}_0(\mathbb{R}^d)\quad\mbox{and}\quad S(\phi_k g_j):=S(\phi_k g_j,\phi_k g_j)\le\tilde{\chi_k}\quad\mbox{for all }k,j\in\mathbb{N},
\end{align*}
where $\tilde{\chi}_{k}:=4\chi_k\Big(\sqrt{S(\chi_k)}+C\,\big(\chi_k+\sqrt{S(\chi_k)}\big)\Big)$ with $\chi_k\in C_0^\infty(\mathbb{R}^d)$ and $C\in (0,\infty)$ as in \cite[Cond.~(Q)]{MaRo00}.
For any function $f\in W_0^{1,2}(\mathbb{R}^d)$ we have $\langle f,\cdot\rangle\in D(\mathcal{E}^{\scriptscriptstyle{\Gamma,\mu}}_{\scriptscriptstyle{\text{env}}})$, since $\mu$ fulfills a Ruelle bound. Hence we can consider
\begin{align*}
S^\Gamma(\langle f,\cdot\rangle):=S^\Gamma(\langle f,\cdot\rangle,\langle f,\cdot\rangle)
\end{align*}
and obtain
\begin{align*}
S^\Gamma(\langle f,\cdot\rangle)=\left\langle\big(\nabla f,\nabla f\big)_{\mathbb{R}^d},\cdot\right\rangle+\big(\langle\nabla f,\cdot\rangle,\langle\nabla f,\cdot\rangle\big)_{\mathbb{R}^d}.
\end{align*}
For $\gamma\in\Gamma$ and $\Lambda:=\text{supp~}f,~f\in W^{1,2}_0(\mathbb{R}^d)$, we have
\begin{multline*}
\big(\langle\nabla f,\gamma\rangle,\langle\nabla f,\gamma\rangle\big)_{\mathbb{R}^d}=\sum_{x\in\gamma}\sum_{y\in\gamma}\big(\nabla f(x),\nabla f(y)\big)_{\mathbb{R}^d}\\
\le\sum_{x\in\gamma}\sum_{y\in\gamma}\sqrt{\big(\nabla f(x),\nabla f(x)\big)_{\mathbb{R}^d}}\cdot\sqrt{\big(\nabla f(y),\nabla f(y)\big)_{\mathbb{R}^d}}\\
=\sum_{x\in\gamma}\sqrt{\big(\nabla f(x),\nabla f(x)\big)_{\mathbb{R}^d}}\cdot\sum_{y\in\gamma}\sqrt{\big(\nabla f(y),\nabla f(y)\big)_{\mathbb{R}^d}}\\
=\left|\gamma_\Lambda\right|^2\left(\frac{1}{\left|\gamma_\Lambda\right|}\sum_{x\in\gamma}\sqrt{\big(\nabla f(x),\nabla f(x)\big)_{\mathbb{R}^d}}\right)^2\le\left|\gamma_\Lambda\right|\sum_{x\in\gamma}\big(\nabla f(x),\nabla f(x)\big)_{\mathbb{R}^d}\\
=\left|\gamma_\Lambda\right|\left\langle\big(\nabla f,\nabla f\big)_{\mathbb{R}^d},\gamma\right\rangle,
\end{multline*}
where we have used Jensen's inequality.
Finally,
\begin{align*}
S^\Gamma(\langle f,\cdot\rangle)\le \left(1+\left|\gamma_\Lambda\right|\right)\cdot\left\langle\big(\nabla f,\nabla f\big)_{\mathbb{R}^d},\cdot\right\rangle=\left(1+\left|\gamma_\Lambda\right|\right)\cdot\left\langle S(f),\cdot\right\rangle,\quad \Lambda=\text{supp~}f,~f\in W^{1,2}_0(\mathbb{R}^d).
\end{align*}
Next we fix a function $\zeta\in C_b^\infty(\mathbb{R})$ such that $0\le \zeta\le 1$ on $[0,\infty)$, $\zeta(t)=t$ on $\left[-\frac{1}{2},\frac{1}{2}\right]$, $\zeta'>0$ and $\zeta''\le 0$. Here $C_b^\infty(\mathbb{R})$ denotes the set of bounded, continuous functions on $\mathbb{R}^d$ which are infinitely often continuously differentiable. Using an argumentation as in \cite[Lemm.~3.2]{RS95} we have that for any fixed $\gamma_0\in\ddot{\Gamma}$ and for any $k,n\in\mathbb{N}$ the restriction to $\Gamma$ of the function
\begin{align*}
\zeta\left(\sup_{j\le n}\big|\langle\phi_k\,g_j,\cdot\rangle-\langle\phi_k\,g_j,\gamma_0\rangle\big|\right)
\end{align*}
belongs to $D(\mathcal{E}^{\scriptscriptstyle{\Gamma,\mu}}_{\scriptscriptstyle{\text{env}}})$. Furthermore, we obtain
\begin{multline*}
S^\Gamma\left(\zeta\left(\sup_{j\le n}\big|\langle\phi_k\,g_j,\cdot\rangle-\langle\phi_k\,g_j,\gamma_0\rangle\big|\right)\right)\\
\le\left(1+N_{B_{k+1}}(\cdot)\right) S^\Gamma_0\left(\zeta\left(\sup_{j\le n}\big|\langle\phi_k\,g_j,\cdot\rangle-\langle\phi_k\,g_j,\gamma_0\rangle\big|\right)\right)\quad\mu\mbox{-a.e.},
\end{multline*}
since $\phi_k g_j,~k,j\in\mathbb{N}$, having support in $B_{k+1}$. Here, as usual, 
$N_B:\Gamma\to \mathbb{N}_0\cup\{+\infty\}$ is given by $N_B(\gamma):=\gamma(B)$, where $B\in\mathcal{B}(\mathbb{R}^d)$.
Due to \cite[(4.7)]{MaRo00} we have
\begin{align*}
S^\Gamma_0\left(\zeta\left(\sup_{j\le n}\big|\langle\phi_k\,g_j,\cdot\rangle-\langle\phi_k\,g_j,\gamma_0\rangle\big|\right)\right)\le\left\langle\tilde{\chi}_k^2,\cdot\right\rangle\quad\mu\mbox{-a.e.}
\end{align*}
Thus
\begin{align}\label{esti}
S^\Gamma\left(\zeta\left(\sup_{j\le n}\big|\langle\phi_k\,g_j,\cdot\rangle-\langle\phi_k\,g_j,\gamma_0\rangle\big|\right)\right)\le \left(1+N_{B_{k+1}}(\cdot)\right)\left\langle\tilde{\chi}_k^2,\cdot\right\rangle\quad\mu\mbox{-a.e.}
\end{align}
For $\gamma,\gamma_0\in\ddot{\Gamma}$ and $k\in\mathbb{N}$ we set 
\begin{align*}
F_k(\gamma,\gamma_0):=\zeta\Big(\sup_{j\in\mathbb{N}}\big|\langle\phi_k g_j,\gamma\rangle-\langle\phi_k g_k,\gamma_0\rangle\big|\Big)
\end{align*}
and for a fixed $\gamma_0\in\ddot{\Gamma}$
\begin{align*}
\zeta\left(\sup_{j\le n}\big|\langle\phi_k\,g_j,\gamma\rangle-\langle\phi_k\,g_j,\gamma_0\rangle\big|\right)\to F_k(\gamma,\gamma_0)\quad\mbox{as }n\to\infty\mbox{ for all }\gamma\in\ddot{\Gamma},
\end{align*}
and in $L^2(\ddot{\Gamma},\mu)$. Hence by (\ref{esti}) and the Banach-Saks theorem, $F_k(\cdot,\gamma_0)\in D(\mathcal{E}^{\scriptscriptstyle{\Gamma,\mu}}_{\scriptscriptstyle{\text{env}}}))$ and
\begin{align}\label{esti2}
S^\Gamma(F_k(\cdot,\gamma_0))\le\left(1+N_{B_{k+1}}(\cdot)\right)\left\langle\tilde{\chi}_{k}^2,\cdot\right\rangle\quad\mu\mbox{-a.e..}
\end{align}
Next let us define
\begin{multline*}
c_k:=\Big(1+\frac{1}{2}\int_{\mathbb{R}^{2d}}1_{B_{k+1}}(x_1)\tilde{\chi}_{k}^2(x_2)\rho_{\mu}^{\scriptscriptstyle{(2)}}(x_1,x_2)\exp(-\phi)(x_1)\exp(-\phi)(x_2)\,dx_1\,dx_2\\
+\int_{\mathbb{R}^d}1_{B_{k+1}}(x_1)\tilde{\chi}_{k}^2(x_1)\rho_{\mu}^{\scriptscriptstyle{(1)}}(x_1)\exp(-\phi)(x_1)\,dx_1\\
+\int_{\mathbb{R}^d}\tilde{\chi}_k^2(x_1)\rho_{\mu}^{\scriptscriptstyle{(1)}}(x_1)\exp(-\phi)(x_1)\,dx_1\Big)^{-\frac{1}{2}}2^{-\frac{k}{2}},\quad k\in\mathbb{N}.
\end{multline*}
Note that since $\mu$ fulfills a Ruelle bound and $\phi$ is bounded from below $(c_k)_{k\in\mathbb{N}}$ is a sequence of positive real numbers converging to $0$ as $k\to\infty$.
For $\gamma_1,\gamma_2\in\ddot{\Gamma}$ we define
\begin{align*}
\bar{\rho}(\gamma_1,\gamma_2):=\sup_{k\in\mathbb{N}} c_k\,F_k(\gamma_1,\gamma_2).
\end{align*}
By \cite[Theo.~3.6]{MaRo00}, $\bar{\rho}$ is a bounded, complete metric on $\ddot{\Gamma}$ generating the vague
topology.
Furthermore,
\begin{multline*}
S^\Gamma(c_k\,F_k(\cdot,\gamma_0))\le 2^{-k}\Big(1+\frac{1}{2}\int_{\mathbb{R}^{2d}}\!\!\!1_{B_{k+1}}(x_1)\tilde{\chi}_{k}^2(x_2)\rho_{\mu}^{\scriptscriptstyle{(2)}}(x_1,x_2)\exp(-\phi)(x_1)\exp(-\phi)(x_2)\,dx_1\,dx_2\\
+\int_{\mathbb{R}^d}1_{B_{k+1}}(x_1)\tilde{\chi}_{k}^2(x_1)\rho_{\mu}^{\scriptscriptstyle{(1)}}(x_1)\exp(-\phi)(x_1)\,dx_1
+\int_{\mathbb{R}^d}\tilde{\chi}_{k}^2(x_1)\rho_{\mu}^{\scriptscriptstyle{(1)}}(x_1)\exp(-\phi)(x_1)\,dx_1\Big)^{-1}\\
\times\left(N_{B_{k+1}}(\cdot)+1\right)\left\langle\tilde{\chi}_{k}^2,\cdot\right\rangle\\
\le\sup_{k\in\mathbb{N}}\Bigg(2^{-k}\Big(1+\frac{1}{2}\int_{\mathbb{R}^{2d}}1_{B_{k+1}}(x_1)\tilde{\chi}_{k}^2(x_2)\rho_{\mu}^{\scriptscriptstyle{(2)}}(x_1,x_2)\exp(-\phi)(x_1)\exp(-\phi)(x_2)\,dx_1\,dx_2\\
+\int_{\mathbb{R}^d}1_{B_{k+1}}(x_1)\tilde{\chi}_{k}^2(x_1)\rho_{\mu}^{\scriptscriptstyle{(1)}}(x_1)\exp(-\phi)(x_1)\,dx_1
+\int_{\mathbb{R}^d}\tilde{\chi}_{k}^2(x_1)\rho_{\mu}^{\scriptscriptstyle{(1)}}(x_1)\exp(-\phi)(x_1)\,dx_1\Big)^{-1}\\
\times\left(N_{B_{k+1}}(\cdot)+1\right)\left\langle\tilde{\chi}_{k}^2,\cdot\right\rangle\Bigg)=:\eta\quad\mu\mbox{-a.e.}
\end{multline*}
by (\ref{esti2}). Thus  by \cite[Lemm.~3.2]{RS95} we have for all $n\in\mathbb{N}$
\begin{align*}
S^\Gamma\left(\sup_{k\le n}c_k F_k(\cdot,\gamma_0)\right)\le\sup_{k\le n}S^\Gamma\left(c_k F_k(\cdot,\gamma_0)\right)\le\sup_{k\in\mathbb{N}}S^\Gamma\left(c_k F_k(\cdot,\gamma_0)\right)\le\eta\quad\mu\mbox{-a.e.}
\end{align*}
But $\sup_{k\le n}c_k\,F_k(\cdot,\gamma_0)\to\bar{\rho}(\cdot,\gamma_0)$ as $n\to\infty$ pointwisely and in $L^2(\ddot{\Gamma},\mu)$. Thus $\bar{\rho}(\cdot,\gamma_0)\in D(\mathcal{E}^{\scriptscriptstyle{\Gamma,\mu}}_{\scriptscriptstyle{\text{env}}})$ and $S^\Gamma(\bar{\rho}(\cdot,\gamma_0))\le\eta$, by the Banach-Saks theorem, since  
\begin{multline*}
\int_\Gamma\eta\,d\mu
\le\sum_{k=1}^\infty 2^{-k}\Big(1+\frac{1}{2}\int_{\mathbb{R}^{2d}}1_{B_{k+1}}(x_1)\tilde{\chi}_{k}^2(x_2)\rho_{\mu}^{\scriptscriptstyle{(2)}}(x_1,x_2)\exp(-\phi)(x_1)\exp(-\phi)(x_2)\,dx_1\,dx_2\\
+\int_{\mathbb{R}^d}1_{B_{k+1}}(x_1)\tilde{\chi}_{k}^2(x_1)\rho_{\mu}^{\scriptscriptstyle{(1)}}(x_1)\exp(-\phi)(x_1)\,dx_1
+\int_{\mathbb{R}^d}\tilde{\chi}_{k}^2(x_1)\rho_{\mu}^{\scriptscriptstyle{(1)}}(x_1)\exp(-\phi)(x_1)\,dx_1\Big)^{-1}\\
\times\Big(\frac{1}{2}\int_{\mathbb{R}^{2d}}1_{B_{k+1}}(x_1)\tilde{\chi}_{k}^2(x_2)\rho_{\mu}^{\scriptscriptstyle{(2)}}(x_1,x_2)\exp(-\phi)(x_1)\exp(-\phi)(x_2)\,dx_1\,dx_2\\
+\int_{\mathbb{R}^d}1_{B_{k+1}}(x_1)\tilde{\chi}_{k}^2(x_1)\rho_{\mu}^{\scriptscriptstyle{(1)}}(x_1)\exp(-\phi)(x_1)\,dx_1
+\int_{\mathbb{R}^d}\tilde{\chi}_{k}^2(x_1)\rho_{\mu}^{\scriptscriptstyle{(1)}}(x_1)\exp(-\phi)(x_1)\,dx_1\Big)<\infty.
\end{multline*}
\end{proof}

\begin{lemma}\label{lemE23}
$\left(\mathcal{E}^{\scriptscriptstyle{\Gamma,\mu}}_{\scriptscriptstyle{\text{env}}},D(\mathcal{E}^{\scriptscriptstyle{\Gamma,\mu}}_{\scriptscriptstyle{\text{env}}})\right)$ is local, i.e., $\mathcal{E}^{\scriptscriptstyle{\Gamma,\mu}}_{\scriptscriptstyle{\text{env}}}(F,G)=0$ provided $F,G\in D(\mathcal{E}^{\scriptscriptstyle{\Gamma,\mu}}_{\scriptscriptstyle{\text{env}}})$ with\\ $\text{supp}(|F|\cdot\mu)\cap\text{supp}(|G|\cdot\mu)=\varnothing$.
\end{lemma}

\begin{proof}
The proof is a  simple modification of the proof of \cite[Prop.~4.12]{MaRo00}, where similar arguments as in the proof of Lemma \ref{lemE22} are used.
\end{proof}

Combining Lemmas \ref{lemE21}, \ref{lemE22} and \ref{lemE23} we obtain

\begin{theorem}\label{thmdirE2}
Suppose that the pair potential $\phi$ satisfies (SS), (I), (LR), (D$\text{L}^\text{2}$) and (LS) and let $\mu\in\mathcal{G}^{\scriptscriptstyle{gc}}_{\scriptscriptstyle{\text{ibp}}}(\Phi_{\scriptscriptstyle{\phi}},z\exp(-\phi)),~0<z<\infty$. Then
$\left(\mathcal{E}^{\scriptscriptstyle{\Gamma,\mu}}_{\scriptscriptstyle{\text{env}}},D(\mathcal{E}^{\scriptscriptstyle{\Gamma,\mu}}_{\scriptscriptstyle{\text{env}}})\right)$ is a local, quasi-regular, symmetric Dirichlet form which is conservative, i.e.,~$1\in D(\mathcal{E}^{\scriptscriptstyle{\Gamma,\mu}}_{\scriptscriptstyle{\text{env}}})$ and $\mathcal{E}^{\scriptscriptstyle{\Gamma,\mu}}_{\scriptscriptstyle{\text{env}}}(1,1)=0$. Its generator, denoted by ${H}^{\scriptscriptstyle{\Gamma,\mu}}_{\scriptscriptstyle{\text{env}}}$, is the Friedrichs' extension of $-L^{\scriptscriptstyle{\Gamma,\mu}}_{\scriptscriptstyle{\text{env}}}$.
\end{theorem}

\begin{theorem}\label{thmexpro2}
Suppose the assumptions of Theorem \ref{thmdirE2}. Then
\begin{enumerate}
\item[(i)]
there exists a conservative diffusion process 
\begin{align*}
\mathbf{{M}}^{\scriptscriptstyle{\Gamma,\mu}}_{\scriptscriptstyle{\text{env}}}=\left(\mathbf{{\Omega}},\mathbf{{F}}^{\scriptscriptstyle{\text{env}}},(\mathbf{{F}}^{\scriptscriptstyle{\text{env}}}_t)_{t\ge 0},(\mathbf{X}^{\scriptscriptstyle{\text{env}}}_t)_{t\ge 0},(\mathbf{{P}}^{\scriptscriptstyle{\text{env}}}_\gamma)_{\gamma\in\ddot{\Gamma}}\right)
\end{align*}
on $\ddot{\Gamma}$ which is properly associated with 
$\left(\mathcal{E}^{\scriptscriptstyle{\Gamma,\mu}}_{\scriptscriptstyle{\text{env}}},D(\mathcal{E}^{\scriptscriptstyle{\Gamma,\mu}}_{\scriptscriptstyle{\text{env}}})\right)$, i.e., for all ($\mu$-versions of) $F\in L^2(\ddot{\Gamma},\mu)$ and all $t>0$ the function
\begin{align*}
\gamma\mapsto p_t^{\scriptscriptstyle{\text{env}}} F(\gamma):=\int_{{\mathbf{\Omega}}}F({\mathbf{X}}^{\scriptscriptstyle{\text{env}}}_t)\,d{\mathbf{{P}}}^{\scriptscriptstyle{\text{env}}}_\gamma,\quad\gamma\in\ddot{\Gamma},
\end{align*}
is an $\mathcal{E}^{\scriptscriptstyle{\Gamma,\mu}}_{\scriptscriptstyle{\text{env}}}$-quasi-continuous version of $\exp(-t {H}^{\scriptscriptstyle{\Gamma,\mu}}_{\scriptscriptstyle{\text{env}}})F$. $\mathbf{{M}}^{\scriptscriptstyle{\Gamma,\mu}}_{\scriptscriptstyle{\text{env}}}$ is up to $\mu$-equivalence unique (cf.~\cite[Chap.~IV,~Sect.~6]{MaRo92}). In particular, $\mathbf{{M}}^{\scriptscriptstyle{\Gamma,\mu}}_{\scriptscriptstyle{\scriptscriptstyle{\text{env}}}}$ is $\mu$-symmetric, i.e.,
\begin{align*}
\int_{\ddot{\Gamma}} G\,p^{\scriptscriptstyle{\text{env}}}_t F\,d\mu(\gamma)=\int_{\ddot{\Gamma}}F\,p^{\scriptscriptstyle{\text{env}}}_t G\,d\mu(\gamma)\quad\mbox{for all }F,G:\ddot{\Gamma}\to\mathbb{R_+},~\mathcal{B}(\ddot{\Gamma})\mbox{-measurable}
\end{align*}
and has $\mu$ as invariant measure.
\item[(ii)]
$\mathbf{{M}}^{\scriptscriptstyle{\Gamma,\mu}}_{\scriptscriptstyle{\text{env}}}$ from (i) is the (up to $\mu$-equivalence,~cf.~\cite[Def.~6.3]{MaRo92}) unique diffusion process having $\mu$ as invariant measure and solving the martingale problem for\\ $\left(-{H}^{\scriptscriptstyle{\Gamma,\mu}}_{\scriptscriptstyle{\text{env}}},D({H}^{\scriptscriptstyle{\Gamma,\mu}}_{\scriptscriptstyle{\text{env}}})\right)$, i.e., for all $G\in D({H}^{\scriptscriptstyle{\Gamma,\mu}}_{\scriptscriptstyle{\text{env}}})\supset\mathcal{F}C_b^\infty(C^\infty_0(\mathbb{R}^d),\Gamma)$
\begin{align*}
\widetilde{G}(\mathbf{X}^{\scriptscriptstyle{\text{env}}}_t)-\widetilde{G}(\mathbf{X}^{\scriptscriptstyle{\text{env}}}_0)+\int_0^t {H}^{\scriptscriptstyle{\Gamma,\mu}}_{\scriptscriptstyle{\text{env}}} G(\mathbf{{X}}^{\scriptscriptstyle{\text{env}}}_t)\,ds,\quad t\ge 0,
\end{align*}
is an $(\mathbf{{F}}_t)_{t\ge 0}$-martingale under $\mathbf{{P}}^{\scriptscriptstyle{\text{env}}}_\gamma$ (hence starting at $\gamma$) for $\mathcal{E}^{\scriptscriptstyle{\Gamma,\mu}}_{\scriptscriptstyle{\text{env}}}$-q.a.~$\gamma\in\ddot{\Gamma}$. (Here $\widetilde{G}$ denotes a quasi-continuous version of $G$, cf.~\cite[Chap.~IV,~Prop.3.3]{MaRo92}.)
\end{enumerate}
\end{theorem}

\begin{proof}
\begin{enumerate}
\item[(i)]
By Theorem \ref{thmdirE2} the proof follows directly from \cite[Chap.~V,~Theo.~1.11]{MaRo92}.
\item[(ii)]
This follows immediately by \cite[Theo.~3.5]{MR1335494}.
\end{enumerate}
\end{proof}

\begin{remark}\label{remexcept2}
For $d\ge 2$ an argumentation as in the proofs of \cite[Prop.~1]{RS98} and \cite[Coro.~1]{RS98} together with a similar argumentation as in the proof of Lemma \ref{lemE22} gives us that under our assumptions the set $\ddot{\Gamma}\setminus\Gamma$ is $\mathcal{E}^{\scriptscriptstyle{\Gamma,\mu}}_{\scriptscriptstyle{\text{env}}}$-exceptional. Therefore, the process $\mathbf{M}^{\scriptscriptstyle{\Gamma,\mu}}_{\scriptscriptstyle{\text{env}}}$ from Theorem \ref{thmexpro2} lives on the smaller space $\Gamma$.
\end{remark}

\subsection{The coupled process}\index{coupled process}

Finally we construct the stochastic process $\mathbf{M}^{\scriptscriptstyle{\mathbb{R}^d\times\Gamma,\hat{\mu}}}_{\scriptscriptstyle{\text{coup}}}$ taking values in $\mathbb{R}^d\times\Gamma$ for $d\ge 2$ (for $d=1$ the process exists only in $\mathbb{R}^d\times\ddot{\Gamma}$), coupling the motion of the tagged particle and the motion of the environment seen from this particle. Therefore, let $\mu\in\mathcal{G}^{\scriptscriptstyle{gc}}_{\scriptscriptstyle{\text{ibp}}}(\Phi_{\scriptscriptstyle{\phi}},z\exp(-\phi)),~0<z<\infty$. As test functions we consider functions $\mathfrak{F}\in C^\infty_0(\mathbb{R}^d)\otimes\mathcal{F}C_b^{\infty}(C^{\infty}_0(\mathbb{R}^d),\Gamma)$. Here $\otimes$ denotes the algebraic tensor product of $C^\infty_0(\mathbb{R}^d)$ and $\mathcal{F}C_b^{\infty}(C^{\infty}_0(\mathbb{R}^d),\Gamma)$. Hence
\begin{align}\label{short}
\mathfrak{F}(\xi,\gamma)=\sum_{k=1}^{m_{\scriptscriptstyle{\mathfrak{F}}}}(f_k\otimes F_k)(\xi,\gamma):=\sum_{k=1}^{m_{\scriptscriptstyle{\mathfrak{F}}}}f_k(\xi)F_k(\gamma),\quad(\xi,\gamma)\in\mathbb{R}^d\times\Gamma,
\end{align}
where $f_k\in C^\infty_0(\mathbb{R}^d)$, $F_k\in\mathcal{F}C_b^{\infty}(C^{\infty}_0(\mathbb{R}^d),\Gamma)$ for $k\in\{1,\ldots,m_{\scriptscriptstyle{\mathfrak{F}}}\}$ and $m_{\scriptscriptstyle{\mathfrak{F}}}\in\mathbb{N}$ depends on $\mathfrak{F}\in C^\infty_0(\mathbb{R}^d)\otimes\mathcal{F}C_b^{\infty}(C^{\infty}_0(\mathbb{R}^d),\Gamma)$. 

As operators on $C^\infty_0(\mathbb{R}^d)\otimes\mathcal{F}C_b^{\infty}(C^{\infty}_0(\mathbb{R}^d),\Gamma)$ we consider
\begin{align}
(&\nabgam -\nabla)\mathfrak{F}(\xi,\gamma):=\sum_{k=1}^{m_{\scriptscriptstyle{\mathfrak{F}}}}\Big(f_k(\xi)\nabgam F_k(\gamma)-\nabla f_k(\xi)F(\gamma)\Big)\quad\mbox{and}\label{sh1}\\
&\nabla^\Gamma \mathfrak{F}(\xi,\gamma):=\sum_{k=1}^{m_{\scriptscriptstyle{\mathfrak{F}}}}f_k(\xi)\nabla^\Gamma F_k(\gamma),\quad\mbox{for }(\xi,\gamma)\in\mathbb{R}^d\times\Gamma,\label{sh2}
\end{align}
where $\mathfrak{F}\in C^\infty_0(\mathbb{R}^d)\otimes\mathcal{F}C_b^{\infty}(C^{\infty}_0(\mathbb{R}^d),\Gamma)$.
\begin{notation}\label{simp}
Since the objects we consider are linear or bilinear, respectively, for simplicity we use
\begin{align*}
&\mathfrak{F}(\xi,\gamma)=f(\xi)\,F(\gamma)\mbox{ instead of (\ref{short})},\\
&(\nabgam -\nabla)\mathfrak{F}(\xi,\gamma)=f(\xi)\,\nabgam F(\gamma)-\nabla f(\xi)\,F(\gamma)\mbox{ instead of (\ref{sh1}) and}\\
&\nabla^\Gamma \mathfrak{F}(\xi,\gamma)=f(\xi)\,\nabla^\Gamma F(\gamma)\mbox{ instead of (\ref{sh2})}.
\end{align*}
\end{notation}
Now we define on $C^\infty_0(\mathbb{R}^d)\otimes\mathcal{F}C_b^{\infty}(C^{\infty}_0(\mathbb{R}^d),\Gamma)$ the following positive definite, symmetric bilinear form:
\begin{multline}\label{equcoup}
\mathcal{E}^{\scriptscriptstyle{\mathbb{R}^d\times\Gamma,\hat{\mu}}}_{\scriptscriptstyle{\text{coup}}}(\mathfrak{F},\mathfrak{G})=\int_{\mathbb{R}^d\times\Gamma}\Big((\nabgam -\nabla)\mathfrak{F}(\xi,\gamma),(\nabgam -\nabla)\mathfrak{G}(\xi,\gamma)\Big)_{\mathbb{R}^d}\\+\Big(\nabla^\Gamma \mathfrak{F}(\xi,\gamma),\nabla^\Gamma \mathfrak{G}(\xi,\gamma)\Big)_{T_\gamma\Gamma}\,d\mu(\gamma)\,d\xi,\\
\quad\mathfrak{F},\mathfrak{G}\in C_0^\infty(\mathbb{R}^d)\otimes\mathcal{F}C_b^{\infty}(C^{\infty}_0(\mathbb{R}^d),\Gamma).
\end{multline}

Using Notation \ref{simp}, (\ref{equcoup}) can be rewritten as
\begin{multline*}
\mathcal{E}^{\scriptscriptstyle{\mathbb{R}^d\times\Gamma,\hat{\mu}}}_{\scriptscriptstyle{\text{coup}}}(\mathfrak{F},\mathfrak{G})=\int_{\mathbb{R}^d\times\Gamma}f(\xi)\,g(\xi)\,\bigg(\Big(\nabla^\Gamma F(\gamma),\nabla^\Gamma G(\gamma)\Big)_{T_\gamma\Gamma}+\Big(\nabgam F(\gamma),\nabgam G(\gamma)\Big)_{\mathbb{R}^d}\bigg)\\
-\Big(\nabgam F(\gamma),\nabla g(\xi)\Big)_{\mathbb{R}^d}f(\xi)\,G(\gamma)-\Big(\nabgam G(\gamma),\nabla f(\xi)\Big)_{\mathbb{R}^d}g(\xi)\,F(\gamma)\\
+F(\gamma)\,G(\gamma) \big(\nabla f(\xi),\nabla g(\xi)\big)_{\mathbb{R}^d}\,d\mu(\gamma)\,d\xi.
\end{multline*}

\begin{theorem}\label{thmexcoup}
Suppose that the pair potential $\phi$ satisfies (SS), (I), (LR), (LS) and (D$\text{L}^{\text{q}}$) for some $q>d$. Furthermore, let $\mu\in\mathcal{G}^{\scriptscriptstyle{gc}}_{\scriptscriptstyle{\text{ibp}}}(\Phi_{\scriptscriptstyle{\phi}},z\exp(-\phi)),~0<z<\infty$. Then\\
$(\mathcal{E}^{\scriptscriptstyle{\mathbb{R}^d\times\Gamma,\hat{\mu}}}_{\scriptscriptstyle{\text{coup}}},C^\infty_0(\mathbb{R}^d)\otimes\mathcal{F}C_b^{\infty,\mu}(C^{\infty}_0(\mathbb{R}^d),\Gamma)$ is closable in $L^2(\mathbb{R}^d\times\Gamma,dx\otimes\mu)$ and its closure\\ $(\mathcal{E}^{\scriptscriptstyle{\mathbb{R}^d\times\Gamma,\hat{\mu}}}_{\scriptscriptstyle{\text{coup}}},D(\mathcal{E}^{\scriptscriptstyle{\mathbb{R}^d\times\Gamma,\hat{\mu}}}_{\scriptscriptstyle{\text{coup}}}))$ is a conservative, local, quasi-regular Dirichlet form on $L^2(\mathbb{R}^d\times \ddot{\Gamma},\mu)$. Moreover,
\begin{align*}
\mathcal{E}^{\scriptscriptstyle{\mathbb{R}^d\times\Gamma,\hat{\mu}}}_{\scriptscriptstyle{\text{coup}}}(\mathfrak{F},\mathfrak{G})=\int_{\mathbb{R}^d\times\Gamma}-L^{\scriptscriptstyle{\mathbb{R}^d\times\Gamma,\hat{\mu}}}_{\scriptscriptstyle{\text{coup}}}\mathfrak{F}\,\mathfrak{G}\,d\mu\,d\xi,
\end{align*}
where
\begin{align*}
L^{\scriptscriptstyle{\mathbb{R}^d\times\Gamma,\hat{\mu}}}_{\scriptscriptstyle{\text{coup}}}\mathfrak{F}(\xi,\gamma)=L^{\scriptscriptstyle{\Gamma,\mu}}_{\scriptscriptstyle{\text{env}}}F(\xi,\gamma)\,f(\xi)\!-2\,\Big(\nabgam F(\gamma),\nabla f(\xi)\Big)_{\mathbb{R}^d}
\!\!\!+\sum_{x\in\gamma}\Big(\nabla\phi(x),\nabla f(\xi)\Big)_{\mathbb{R}^d}\!\!\!+\Delta f(\xi)\,F(\gamma).
\end{align*}
The generator of $(\mathcal{E}^{\scriptscriptstyle{\mathbb{R}^d\times\Gamma,\hat{\mu}}}_{\scriptscriptstyle{\text{coup}}},D(\mathcal{E}^{\scriptscriptstyle{\mathbb{R}^d\times\Gamma,\hat{\mu}}}_{\scriptscriptstyle{\text{coup}}}))$, denoted by ${H}^{\scriptscriptstyle{\mathbb{R}^d\times\Gamma,\hat{\mu}}}_{\scriptscriptstyle{\text{coup}}}$, is the Friedrichs' extension of $-L^{\scriptscriptstyle{\mathbb{R}^d\times\Gamma,\hat{\mu}}}_{\scriptscriptstyle{\text{coup}}}$.
\end{theorem}

\begin{proof}
Applying Fubini's theorem then carrying out an integration by parts, we obtain 
\begin{multline*}
\mathcal{E}^{\scriptscriptstyle{\mathbb{R}^d\times\Gamma,\hat{\mu}}}_{\scriptscriptstyle{\text{coup}}}(\mathfrak{F},\mathfrak{G})=\int_{\mathbb{R}^d\times\Gamma}f(\xi)\,g(\xi)\,\bigg(\Big(\nabla^\Gamma F(\gamma),\nabla^\Gamma G(\gamma)\Big)_{T_\gamma\Gamma}+\Big(\nabgam F(\gamma),\nabgam G(\gamma)\Big)_{\mathbb{R}^d}\bigg)\\
-\Big(\nabgam F(\gamma),\nabla g(\xi)\Big)_{\mathbb{R}^d}f(\xi)\,G(\gamma)-\Big(\nabgam G(\gamma),\nabla f(\xi)\Big)_{\mathbb{R}^d}g(\xi)\,F(\gamma)\\
+F(\gamma)\,G(\gamma)\big(\nabla f(\xi),\nabla g(\xi)\big)_{\mathbb{R}^d}\,d\mu(\gamma)\,d\xi\\
=\int_{\mathbb{R}^d}\int_\Gamma -L^{\scriptscriptstyle{\Gamma,\mu}}_{\scriptscriptstyle{\text{env}}}F(\gamma)\,f(\xi)\,G(\gamma)\,g(\xi)\,d\mu(\gamma)\,dx
+\int_{\mathbb{R}^d}\int_\Gamma\Big(\nabgam F(\gamma),\nabla f(\xi)\Big)_{\mathbb{R}^d}\,G(\gamma)\,g(\xi)\,d\mu(\gamma)\,d\xi\\
+\int_{\mathbb{R}^d}\int_\Gamma\Bigg(\Big(\nabgam F(\gamma),\nabla f(\xi)\Big)_{\mathbb{R}^d}-\sum_{x\in\gamma}\Big(\nabla\phi(x),\nabla f(\xi)\Big)_{\mathbb{R}^d}F(\gamma)\Bigg)\,G(\gamma)\,g(\xi)\,d\mu(\gamma)\,d\xi
\\-\int_{\mathbb{R}^d}\int_\Gamma\Delta f(\xi) F(\gamma)\, G(\gamma)\,g(\xi)\,d\mu(\gamma)\,d\xi\\
=\!\!\!\int_{\mathbb{R}^d}\int_\Gamma\Bigg(-L^{\scriptscriptstyle{\Gamma,\mu}}_{\scriptscriptstyle{\text{env}}}F(\gamma)\,f(\xi)+2\,\Big(\nabgam F(\gamma),\nabla f(\xi)\Big)_{\mathbb{R}^d}\\
-\sum_{x\in\gamma}\Big(\nabla\phi(x),\nabla f(\xi)\Big)_{\mathbb{R}^d}\,F(\gamma)
-\Delta f(\xi)\,F(\gamma)\Bigg)G(\gamma)\,g(\xi)\,d\mu(\gamma)\,d\xi.
\end{multline*}
Thus we have closability. Since the operator $\nabla_\gamma^\Gamma-\nabla$ fulfills a chain rule on $C^\infty_0(\mathbb{R}^d)\otimes\mathcal{F}C_b^{\infty}(C^{\infty}_0(\mathbb{R}^d),\Gamma)$, the Dirichlet property follows. Furthermore, $\nabla_\gamma^\Gamma-\nabla$ satisfies the product rule for bounded functions in $D(\mathcal{E}^{\scriptscriptstyle{\mathbb{R}^d\times\Gamma,\hat{\mu}}}_{\scriptscriptstyle{\text{coup}}})$ then as shown in \cite[Chap.~V, Exam.~1.12(ii)]{MaRo92}, $(\mathcal{E}^{\scriptscriptstyle{\mathbb{R}^d\times\Gamma,\hat{\mu}}}_{\scriptscriptstyle{\text{coup}}},D(\mathcal{E}^{\scriptscriptstyle{\mathbb{R}^d\times\Gamma,\mu}}_{\scriptscriptstyle{\text{coup}}}))$ is local.

Quasi-regularity can be shown as follows. We denote by $(\mathcal{E},D(\mathcal{E}))$ the classical gradient Dirichlet form on $L^2(\mathbb{R}^d,dx)$. We know that both $(\mathcal{E}^{\scriptscriptstyle{\Gamma,{\mu}}}_{\scriptscriptstyle{\text{env}}},D(\mathcal{E}^{\scriptscriptstyle{\Gamma,{\mu}}}_{\scriptscriptstyle{\text{env}}}))$ and $(\mathcal{E},D(\mathcal{E}))$ are quasi-regular. By $(E_k)_{k\in\mathbb{N}}$ we denote the $\mathcal{E}^{\scriptscriptstyle{\Gamma,{\mu}}}_{\scriptscriptstyle{\text{env}}}$-nest of compact sets in $\ddot{\Gamma}$. An $\mathcal{E}$-nest of compact sets in $\mathbb{R}^d$ is given by $(\overline{\Lambda_k})_{k\in\mathbb{N}}$. Hence $(F_k)_{k\in\mathbb{N}}$, where $F_k:=\overline{\Lambda}_k\times E_k$, is an exhausting sequence of compact sets in $\mathbb{R}^d\times\ddot{\Gamma}$. One easily shows that $C_0^\infty(\mathbb{R}^d)\otimes\bigcup_{k\ge 1}D(\mathcal{E}^{\scriptscriptstyle{\Gamma,{\mu}}}_{\scriptscriptstyle{\text{env}}})_{E_k} \subset  D(\mathcal{E}^{\scriptscriptstyle{\mathbb{R}^d\times\Gamma,\mu}}_{\scriptscriptstyle{\text{coup}}})$. Then using that $\bigcup_{k\ge 1}D(\mathcal{E}^{\scriptscriptstyle{\Gamma,{\mu}}}_{\scriptscriptstyle{\text{env}}})_{E_k}$ is dense in $D(\mathcal{E}^{\scriptscriptstyle{\Gamma,{\mu}}}_{\scriptscriptstyle{\text{env}}})$ with respect to $\sqrt{\mathcal{E}^{\scriptscriptstyle{\Gamma,{\mu}}}_{\scriptscriptstyle{\text{env}}}}$ and that $C_0^\infty(\mathbb{R}^d)$ is dense in $D(\mathcal{E})$ with respect to $\sqrt{\mathcal{E}}$ we can easily show that $C_0^\infty(\mathbb{R}^d)\otimes\bigcup_{k\ge 1}D(\mathcal{E}^{\scriptscriptstyle{\Gamma,{\mu}}}_{\scriptscriptstyle{\text{env}}})_{E_k}$ is dense, first in $C^\infty_0(\mathbb{R}^d)\otimes\mathcal{F}C_b^{\infty}(C^{\infty}_0(\mathbb{R}^d),\Gamma)$, hence also in $D(\mathcal{E}^{\scriptscriptstyle{\mathbb{R}^d\times\Gamma,\hat{\mu}}}_{\scriptscriptstyle{\text{coup}}})$ with respect to $\sqrt{\mathcal{E}^{\scriptscriptstyle{\mathbb{R}^d\times\Gamma,\hat{\mu}}}_{\scriptscriptstyle{\text{coup}}}}$. Thus $(F_k)_{k\in\mathbb{N}}$ is an $\mathcal{E}^{\scriptscriptstyle{\mathbb{R}^d\times\Gamma,\hat{\mu}}}_{\scriptscriptstyle{\text{coup}}}$-nest of compact sets. All further properties necessary to have quasi-regularity are clear due to quasi-regularity of $(\mathcal{E}^{\scriptscriptstyle{\Gamma,{\mu}}}_{\scriptscriptstyle{\text{env}}},D(\mathcal{E}^{\scriptscriptstyle{\Gamma,{\mu}}}_{\scriptscriptstyle{\text{env}}}))$ and $(\mathcal{E},D(\mathcal{E}))$, respectively.    

Finally, we prove the conservativity of $(\mathcal{E}^{\scriptscriptstyle{\mathbb{R}^d\times\Gamma,\hat{\mu}}}_{\scriptscriptstyle{\text{coup}}},D(\mathcal{E}^{\scriptscriptstyle{\mathbb{R}^d\times\Gamma,\mu}}_{\scriptscriptstyle{\text{coup}}}))$. I.e., we have to show that $T^{\scriptscriptstyle{\text{coup}}}_t(1\otimes 1)=1$, with $(T^{\scriptscriptstyle{\text{coup}}}_t)_{t\ge 0}$ the $L^\infty$-contraction semigroup corresponding to\\ $(\exp(-tH^{\scriptscriptstyle{\text{coup}}}))_{t\ge 0}$. Therefore, we denote by $(G^{\scriptscriptstyle{\text{coup}}}_\alpha)_{\alpha>0}$ the resolvent corresponding to $\mathcal{E}^{\scriptscriptstyle{\mathbb{R}^d\times\Gamma,\hat{\mu}}}_{\scriptscriptstyle{\text{coup}}}$ and prove at first that
\begin{align*}
_{\scriptscriptstyle{L^1}}(\mathfrak{F},G_1^{\scriptscriptstyle{\text{coup}}}(1\otimes 1))_{\scriptscriptstyle{L^\infty}}={_{\scriptscriptstyle{L^1}}(}\mathfrak{F}, 1)_{\scriptscriptstyle{L^\infty}}\quad\mbox{for all }\mathfrak{F}\in L^1(\mathbb{R}^d\times\Gamma,\hat{\mu})\cap L^{\infty}(\mathbb{R}^d\times\Gamma,\hat{\mu}).
\end{align*}
Here $_{\scriptscriptstyle{L^1}}(\cdot,\cdot)_{\scriptscriptstyle{L^\infty}}$ denotes the dual pairing between the spaces $L^1(\mathbb{R}^d\times\Gamma,\hat{\mu})$ and $L^\infty(\mathbb{R}^d\times\Gamma,\hat{\mu})$.
In order to show this we choose $f_0:\mathbb{R}^d\to\mathbb{R}$ infinitely often differentiable such that $f_0(x)=1$ for $x\in[-1,1]^d$ and $f_0(x)=0$ for $x\in\mathbb{R}^d\setminus[-3,3]^d$. For $k\in\mathbb{N}$ we define $f_k(x)=f_0(k^{-1}x),~x\in\mathbb{R}^d$. Then for any $q>d$ we have
\begin{align}\label{conv}
\Vert \nabla f_k\Vert_{\scriptscriptstyle{L^p(\mathbb{R}^d)}}\to 0\quad\mbox{and}\quad\Vert \Delta f_k\Vert_{\scriptscriptstyle{L^p(\mathbb{R}^d)}}\to 0\quad\mbox{as }k\to\infty.
\end{align}
It holds that $f_k\otimes 1\in C_0^\infty(\mathbb{R}^d)\otimes\mathcal{F}C_b^{\infty}(C^{\infty}_0(\mathbb{R}^d),\Gamma)$ and
\begin{multline*}
_{\scriptscriptstyle{L^1}}(\mathfrak{F},G_1^{\scriptscriptstyle{\text{coup}}}(1\otimes 1))_{\scriptscriptstyle{L^\infty}}=\lim_{k\to\infty}(\mathfrak{F},G_1^{\scriptscriptstyle{\text{coup}}}(f_k\otimes 1))_{\scriptscriptstyle{L^2}(\mathbb{R}^d\times\Gamma,\hat{\mu})}\\
=\lim_{k\to\infty}((1-L^{\scriptscriptstyle{\Gamma,\hat{\mu}}}_{\scriptscriptstyle{coup}})G_1^{\scriptscriptstyle{\text{coup}}}\mathfrak{F},G_1^{\scriptscriptstyle{\text{coup}}}(f_k\otimes 1))_{\scriptscriptstyle{L^2}(\mathbb{R}^d\times\Gamma,\hat{\mu})}\\
=\lim_{k\to\infty}(G_1^{\scriptscriptstyle{\text{coup}}}\mathfrak{F},f_k\otimes 1)_{\scriptscriptstyle{L^2}(\mathbb{R}^d\times\Gamma,\hat{\mu})}\\
=\lim_{k\to\infty}((1-L^{\scriptscriptstyle{\Gamma,\hat{\mu}}}_{\scriptscriptstyle{coup}}+L^{\scriptscriptstyle{\Gamma,\hat{\mu}}}_{\scriptscriptstyle{coup}})G_1^{\scriptscriptstyle{\text{coup}}}\mathfrak{F},f_k\otimes 1)_{\scriptscriptstyle{L^2}(\mathbb{R}^d\times\Gamma,\hat{\mu})}\\
={_{\scriptscriptstyle{L^1}}}(\mathfrak{F},1\otimes 1)_{\scriptscriptstyle{L^\infty}}+\lim_{k\to\infty}(L^{\scriptscriptstyle{\Gamma,\hat{\mu}}}_{\scriptscriptstyle{coup}}G_1^{\scriptscriptstyle{\text{coup}}}\mathfrak{F},f_k\otimes 1)_{\scriptscriptstyle{L^2(\mathbb{R}^d\times\Gamma,\hat{\mu})}}\\
={_{\scriptscriptstyle{L^1}}}(\mathfrak{F},1)_{\scriptscriptstyle{L^\infty}}+\lim_{k\to\infty}(G_1^{\scriptscriptstyle{\text{coup}}}\mathfrak{F},L^{\scriptscriptstyle{\Gamma,\hat{\mu}}}_{\scriptscriptstyle{coup}}(f_k\otimes 1))_{\scriptscriptstyle{L^2(\mathbb{R}^d\times\Gamma,\hat{\mu})}}\\
={_{\scriptscriptstyle{L^1}}}(\mathfrak{F},1)_{\scriptscriptstyle{L^\infty}}+\lim_{k\to\infty}(G_1^{\scriptscriptstyle{\text{coup}}}\mathfrak{F},\Delta f_k+\nabla f_k\nabla^\Gamma_\gamma\phi)_{\scriptscriptstyle{L^2(\mathbb{R}^d\times\Gamma,\hat{\mu})}}.
\end{multline*}
Since
\begin{align*}
(G_1^{\scriptscriptstyle{\text{coup}}}\mathfrak{F},\Delta f_k+\nabla f_k\nabla^\Gamma_\gamma\phi)_{\scriptscriptstyle{L^2(\mathbb{R}^d\times\Gamma,\hat{\mu})}}\le\Vert 1_{\{\nabla f_k\not=0\}}G_1^{\scriptscriptstyle{\text{coup}}}\mathfrak{F}\Vert_{\scriptscriptstyle{L^p(\mathbb{R}^d\times\Gamma,\hat{\mu})}}\Vert \Delta f_k+\nabla f_k\nabla^\Gamma_\gamma\phi\Vert_{\scriptscriptstyle{L^q(\mathbb{R}^d\times\Gamma,\hat{\mu})}}<\infty,
\end{align*}
by (D$\text{L}^{\text{q}}$), we obtain that $\lim_{k\to\infty}(G_1^{\scriptscriptstyle{\text{coup}}}\mathfrak{F},\Delta f_k+\nabla f_k\nabla^\Gamma_\gamma\phi)_{L^2(\mathbb{R}^d\times\Gamma,\hat{\mu})}=0$ by (\ref{conv}).
Hence $_{\scriptscriptstyle{L^1}}(\mathfrak{F},G_1^{\scriptscriptstyle{\text{coup}}}1\otimes 1)_{\scriptscriptstyle{L^\infty}}={_{\scriptscriptstyle{L^1}}}(\mathfrak{F}, 1)_{\scriptscriptstyle{L^\infty}}$ for all $\mathfrak{F}\in L^1(\mathbb{R}^d\times\Gamma,\hat{\mu})\cap L^{\infty}(\mathbb{R}^d\times\Gamma,\hat{\mu})$. Now using the relation between resolvents and semigroups via the Laplace transform together with the Hahn-Banach theorem we obtain conservativity.

\end{proof}

\begin{theorem}\label{thmexprocoup}
Suppose the assumptions of Theorem \ref{thmexcoup}. 
\begin{enumerate}
\item[(i)]
there exists a conservative diffusion process 
\begin{align*}
\mathbf{{M}}^{\scriptscriptstyle{\mathbb{R}^d\times\Gamma,\hat{\mu}}}_{\scriptscriptstyle{\text{coup}}}=\left({\mathbf{{\Omega}}}^{\scriptscriptstyle{\text{coup}}},{\mathbf{{F}}}^{\scriptscriptstyle{\text{coup}}},({\mathbf{{F}}}^{\scriptscriptstyle{\text{coup}}}_t)_{t\ge 0},({\mathbf{X}}^{\scriptscriptstyle{\text{coup}}}_t)_{t\ge 0},({\mathbf{{P}}}^{\scriptscriptstyle{\text{coup}}}_{(x,\gamma)})_{(x,\gamma)\in\mathbb{R}^d\times\ddot{\Gamma}}\right)
\end{align*}
on $\mathbb{R}^d\times\ddot{\Gamma}$ which is properly associated with 
$\left(\mathcal{E}^{\scriptscriptstyle{\mathbb{R}^d\times\Gamma,\hat{\mu}}}_{\scriptscriptstyle{\text{coup}}},D(\mathcal{E}^{\scriptscriptstyle{\mathbb{R}^d\times\Gamma,\hat{\mu}}}_{\scriptscriptstyle{\text{coup}}})\right)$, i.e., for all ($\hat{\mu}$-versions of) $\mathfrak{F}\in L^2(\mathbb{R}^d\times\ddot{\Gamma},\hat{\mu})$ and all $t>0$ the function
\begin{align*}
(\xi,\gamma)\mapsto p_t^{\scriptscriptstyle{\text{coup}}} \mathfrak{F}(\xi,\gamma):=\int_{{\mathbf{\Omega}}}\mathfrak{F}({\mathbf{X}}^{\scriptscriptstyle{\text{coup}}}_t)\,d{\mathbf{{P}}}^{\scriptscriptstyle{\text{coup}}}_{(\xi,\gamma)},\quad(\xi,\gamma)\in\mathbb{R}^d\times\ddot{\Gamma},
\end{align*}
is an $\mathcal{E}^{\scriptscriptstyle{\Gamma,\hat{\mu}}}_{\scriptscriptstyle{\text{coup}}}$-quasi-continuous version of $\exp(-t {H}^{\scriptscriptstyle{\Gamma,\hat{\mu}}}_{\scriptscriptstyle{\text{coup}}})\mathfrak{F}$. 
$\mathbf{{M}}^{\scriptscriptstyle{\Gamma,\hat{\mu}}}_{\scriptscriptstyle{\text{coup}}}$ is up to $\hat{\mu}$-equivalence unique (cf.~\cite[Chap.~IV,~Sect.~6]{MaRo92}).
In particular, $\mathbf{{M}}^{\scriptscriptstyle{\mathbb{R}^d\times\Gamma,\hat{\mu}}}_{\scriptscriptstyle{\text{coup}}}$ is $d\xi\otimes\mu$-symmetric, i.e.,
\begin{align*}
\int_{\mathbb{R}^d\times\ddot{\Gamma}} \mathfrak{G}\,p^{\scriptscriptstyle{\text{coup}}}_t \mathfrak{F}\,d\hat{\mu}=\int_{\mathbb{R}^d\times\ddot{\Gamma}}\mathfrak{F}\,p^{\scriptscriptstyle{\text{coup}}}_t \mathfrak{G}\,d\hat{\mu}
\end{align*}
for all $\mathfrak{F},\mathfrak{G}:\mathbb{R}^d\times\ddot{\Gamma}\to\mathbb{R_+},~\mathcal{B}(\mathbb{R}^d\times\ddot{\Gamma})${-measurable} and has $\hat{\mu}=d\xi\otimes\mu$ as invariant measure.
\item[(ii)]
$\mathbf{{M}}^{\scriptscriptstyle{\mathbb{R}^d\times\Gamma,\hat{\mu}}}_{\scriptscriptstyle{\text{coup}}}$ from (i) solves the martingale problem for $\left(-{H}^{\scriptscriptstyle{\mathbb{R}^d\times\Gamma,\hat{\mu}}}_{\scriptscriptstyle{\text{coup}}},D({H}^{\scriptscriptstyle{\mathbb{R}^d\times\Gamma,\hat{\mu}}}_{\scriptscriptstyle{\text{coup}}})\right)$,
i.e., for all $\mathfrak{G}\in D({H}^{\scriptscriptstyle{\mathbb{R}^d\times\Gamma,\hat{\mu}}}_{\scriptscriptstyle{\text{coup}}})\supset C_0^\infty(\mathbb{R}^d)\otimes\mathcal{F}C_b^\infty(C^\infty_0(\mathbb{R}^d),\Gamma)$
\begin{align*}
\widetilde{\mathfrak{G}}(\mathbf{X}^{\scriptscriptstyle{\text{coup}}}_t)-\widetilde{\mathfrak{G}}(\mathbf{X}^{\scriptscriptstyle{\text{coup}}}_0)+\int_0^t {H}^{\scriptscriptstyle{\Gamma,\hat{\mu}}}_{\scriptscriptstyle{\text{coup}}} G(\mathbf{{X}}^{\scriptscriptstyle{\text{coup}}}_t)\,ds,\quad t\ge 0,
\end{align*}
is an $(\mathbf{{F}}^{\scriptscriptstyle{coup}}_t)_{t\ge 0}$-martingale under $\mathbf{{P}}^{\scriptscriptstyle{\text{coup}}}_{(\xi,\gamma)}$ (hence starting at $(\xi,\gamma)$) for $\mathcal{E}^{\scriptscriptstyle{\Gamma,\hat{\mu}}}_{\scriptscriptstyle{\text{coup}}}$-q.a.~$(\xi,\gamma)\in\mathbb{R}^d\times\ddot{\Gamma}$.\\
(Here $\widetilde{\mathfrak{G}}$ denotes a quasi-continuous version of $\mathfrak{G}$, cf.~\cite[Chap.~IV,~Prop.3.3]{MaRo92}.)
\end{enumerate}
\end{theorem}

\begin{proof}
\begin{enumerate}
\item[(i)]
By Theorem \ref{thmexcoup} the proof follows directly from \cite[Chap.~V,~Theo.~1.11]{MaRo92}.
\item[(ii)]
This follows immediately by \cite[Theo.~3.5]{MR1335494}.
\end{enumerate}
\end{proof}

\begin{remark}\label{remextppro}\index{tagged particle process}
\begin{enumerate}
\item[(i)]
As before we obtain a diffusion process on $\mathbb{R}^d\times\Gamma$ for $d\ge 2$.
\item[(ii)]
To get the tagged particle process we do a projection on the first component of $\mathbf{M}^{\scriptscriptstyle{\mathbb{R}^d\times\Gamma,\hat{\mu}}}_{\scriptscriptstyle{\text{coup}}}$ taking values in $\mathbb{R}^d$.
\item[(iii)]
Note that in general the tagged particle process no longer is a Markov process.
\end{enumerate}
\end{remark}

\end{document}